\documentclass[12pt]{amsart}
\usepackage{txfonts}
\usepackage{amssymb}
\usepackage[mathscr]{eucal}
\usepackage{amsmath}
\usepackage{amscd}
\usepackage[dvips]{color}
\usepackage{multicol}
\usepackage[all]{xy}
\usepackage{graphicx}
\usepackage{color}
\usepackage{colordvi}
\usepackage{xspace}
\usepackage[colorlinks,final,backref=page,hyperindex,hypertex]{hyperref}
\usepackage[normalem]{ulem} 
\usepackage{tikz}
\usepackage[active]{srcltx}

\begin{document}

\newcommand{\nc}{\newcommand}
\newcommand{\delete}[1]{}

\nc{\mlabel}[1]{\label{#1}}  
\nc{\mcite}[1]{\cite{#1}}  
\nc{\mref}[1]{\ref{#1}}  
\nc{\mbibitem}[1]{\bibitem{#1}} 

\newtheorem{theorem}{Theorem}[section]
\newtheorem{thm}[theorem]{Theorem}
\newtheorem{prop}[theorem]{Proposition}
\newtheorem{lemma}[theorem]{Lemma}
\newtheorem{coro}[theorem]{Corollary}
\newtheorem{cor}[theorem]{Corollary}
\newtheorem{prop-def}{Proposition-Definition}[section]
\newtheorem{claim}{Claim}[section]
\newtheorem{propprop}{Proposed Proposition}[section]
\newtheorem{conjecture}[theorem]{Conjecture}
\newtheorem{assumption}{Assumption}
\newtheorem{condition}[theorem]{Assumption}
\newtheorem{question}[theorem]{Question}
\theoremstyle{definition}
\newtheorem{defn}[theorem]{Definition}
\newtheorem{exam}[theorem]{Example}
\newtheorem{remark}[theorem]{Remark}
\newtheorem{ex}[theorem]{Example}
\newtheorem{conv}[theorem]{Convention}
\renewcommand{\labelenumi}{{\rm(\alph{enumi})}}
\renewcommand{\theenumi}{\alph{enumi}}
\nc{\g}{\mathfrak{g}}
\nc{\U}{\mathfrak{U}}
\def\shu{\joinrel{\!\scriptstyle\amalg\hskip -3.1pt\amalg}\,}
\nc{\Ima}{\operatorname{Im}}            
\nc{\Dom}{\operatorname{Dom}}      
\nc{\Diff}{\operatorname{Diff}}           
\nc{\End}{\operatorname{End}}        
\nc{\Id}{\operatorname{Id}}                 
\nc{\Isom}{\operatorname{Isom}}      
\nc{\Ker}{\operatorname{Ker}}           
\nc{\Lin}{\operatorname{Lin}}             
\nc{\Res}{\operatorname{Res}}         
\nc{\spec}{\operatorname{sp}}           
\nc{\supp}{\operatorname{supp}}      
\nc{\Tr}{\operatorname{Tr}}                 
\nc{\Vol}{\operatorname{Vol}}            
\nc{\sign}{\operatorname{sign}}         
\nc{\id}{\operatorname{id}}
\nc{\lin}{\operatorname{lin}}
\nc{\I}{J}
\nc{\cala}{\mathcal{A}}

\nc{\ot}{\otimes}
\nc{\bfk}{\mathbf{k}}
\nc{\wvec}[2]{{\scriptsize{\big [ \!\!
    \begin{array}{c} #1 \\ #2 \end{array} \!\! \big ]}}}

\nc{\zb}[1]{\textcolor{blue}{Bin: #1}}
\nc{\li}[1]{\textcolor{red}{#1}}
\nc{\sy}[1]{\textcolor{purple}{  #1}}

\newcommand{\Z}{\mathbb{Z}}
\newcommand{\ZZ}{\mathbb{Z}}
\newcommand{\Q}{\mathbb{Q}}               
\newcommand{\QQ}{\mathbb{Q}}               
\newcommand{\R}{\mathbb{R}}
\newcommand{\K}{\mathbb{K}}               
\newcommand{\RR}{\mathbb{R}}               
\newcommand{\coof}{L}           
\newcommand{\coef}{\QQ}
\newcommand{\p}{\partial}         
\nc{\cone}[1]{\langle #1\rangle}
\nc{\cl}{c}                 
\nc{\op}{o}                 
\nc{\ccone}[1]{\langle #1\rangle^\cl}
\nc{\ocone}[1]{\langle #1\rangle^\op}
\nc{\cc}{\mathfrak{C}}      
\nc{\dcc}{\mathfrak{DC}}    
\nc{\oc}{\mathcal{C}^o}      
\nc{\dsmc}{\dcc }  
\nc{\csup}{{^\ast}}
\nc{\bs}{\check{S}\,}

 \nc {\linf}{{\rm lin} (F)^\perp}
\nc{\dirlim}{\displaystyle{\lim_{\longrightarrow}}\,}
\nc{\coalg}{\mathbf{C}}
\nc{\barot}{{\otimes}}

\newcommand{\one}{\mbox{$1 \hspace{-1.0mm} {\bf l}$}}
\newcommand{\A}{\mathcal{A}}              
\newcommand{\Abb}{\mathbb{A}}          
\renewcommand{\a}{\alpha}                    
\renewcommand{\b}{\beta}                       

\newcommand{\B}{\mathcal{B}}              
\newcommand{\C}{\mathbb{C}}
 \newcommand{\calm}{{\mathcal M}}

\newcommand{\CC}{\mathcal{C}}           
\newcommand{\CR}{\mathcal{R}}           
\newcommand{\D}{\mathbb{D}}               
\newcommand{\del}{\partial}                    
\newcommand{\DD}{\mathcal{D}}           
\newcommand{\Dslash}{{D\mkern-11.5mu/\,}} 
\newcommand{\e}{\varepsilon}            
\newcommand{\F}{\mathcal{F}}                
\newcommand{\Ga}{\Gamma}                  
\newcommand{\ga}{\gamma}                   
\renewcommand{\H}{\mathcal{H}}           
\newcommand{\half}{{\mathchoice{\thalf}{\thalf}{\shalf}{\shalf}}}
\newcommand{\hideqed}{\renewcommand{\qed}{}} 
 \renewcommand{\L}{\mathcal{L}}          
\newcommand{\la}{\lambda}                   
\newcommand{\<}{\langle}
\renewcommand{\>}{\rangle}
\newcommand{\M}{\mathcal{M}}            
\newcommand{\Mop}{\star}                     
\newcommand{\N}{\mathbb{N}}             
\newcommand{\norm}[1]{\left\lVert#1\right\rVert}    
\newcommand{\norminf}[1]{\left\lVert#1\right\rVert_\infty} 
\newcommand{\om}{\omega}                 
\newcommand{\Om}{\Omega}                
\newcommand{\ol}{\\widetilde}                  
\newcommand{\OO}{\mathcal{O}}          
\newcommand{\ovc}[1]{\overset{\circ}{#1}}
\newcommand{\ox}{\otimes}                    
\newcommand{\pa}{\partial}
\newcommand{\piso}[1]{\lfloor#1\rfloor} 

\newcommand{\rad}{{\mathbf r}}
\newcommand{\sepword}[1]{\quad\mbox{#1}\quad} 
\newcommand{\set}[1]{\{\,#1\,\}}               
\newcommand{\shalf}{{\scriptstyle\frac{1}{2}}} 
\newcommand{\slim}{\mathop{\mathrm{s\mbox{-}lim}}} 
\renewcommand{\SS}{\mathcal{S}}        
\newcommand{\Sp}{{\rm Sp}}
\newcommand{\sg}{\sigma}                              
\newcommand{\T}{\mathbb{T}}                
\newcommand{\tG}{\widetilde{G}}           
\newcommand{\thalf}{\tfrac{1}{2}}            
\newcommand{\Th}{\Theta}
\renewcommand{\th}{\theta}
\newcommand{\tri}{\Delta}                        
\newcommand{\Trw}{\Tr_\omega}           
\newcommand{\UU}{\mathcal{U}}              
\newcommand{\Afr}{\mathfrak{A}}           
\newcommand{\vf}{\varphi}                       
\newcommand{\x}{\times}                          
\newcommand{\wh}{\widehat}                  
\newcommand{\wt}{\widetilde}                 
\newcommand{\ul}[1]{\underline{#1}}             
\renewcommand{\.}{\cdot}                          
\renewcommand{\:}{\colon}                       
\newcommand{\comment}[1]{\textsf{#1}}
\renewcommand{\g}{\mathfrak{g}}
\renewcommand{\U}{\mathfrak{U}}
\def\shu{\joinrel{\!\scriptstyle\amalg\hskip -3.1pt\amalg}\,}

\nc{\calb}{\mathcal{B}}
\nc{\calc}{\mathcal{C}}
\nc{\calh}{\mathcal{H}}
\nc{\deff}{K}
\nc{\cali}{\mathcal{I}}
\nc{\calp}{\mathcal{P}}
\nc{\cals}{\mathcal{S}}
\nc{\vep}{\varepsilon}
\nc {\ltcone}{lattice cone}
\nc {\lC}{(C, \Lambda _C)}

\title[Counting an infinite number of points and renormalization]{Counting  an infinite number of points: a testing ground for renormalization methods}

\author{Li Guo}
\address{Department of Mathematics and Computer Science,
         Rutgers University,
         Newark, NJ 07102, USA}
\email{liguo@rutgers.edu}

\author{Sylvie Paycha}
\address{Institute of Mathematics,
University of Potsdam,
Am Neuen Palais 10,
D-14469 Potsdam, Germany}
\email{paycha@math.uni-potsdam.de}

\author{Bin Zhang}
\address{Yangtze Center of Mathematics, Department of Mathematics,
Sichuan University, Chengdu, 610064, P. R. China}
\email{zhangbin@scu.edu.cn}

\date{\today}

\begin{abstract}  This is a leisurely introductory account addressed to  non-experts  and  based on previous work by the authors,  on how methods borrowed from physics can be used to "count" an infinite number of points.
We begin with the classical case of counting integer points on the non-negative real axis and the classical Euler-Maclaurin formula. As an intermediate stage, we count integer points on product cones where the roles played by the coalgebra and the algebraic Birkhoff factorization can be appreciated in a relatively simple setting. We then consider the general case of (lattice)  cones   for which we introduce a conilpotent coalgebra of cones, with applications to renormalization of conical zeta values. When evaluated at zero arguments conical zeta functions indeed "count" integer points on cones.
\end{abstract}

\subjclass[2010]{11M32, 11H06, 16H15, 52C07, 52B20, 65B15, 81T15}

\keywords{cones, coalgebra, renormalization, Birkhoff decomposition, Euler-Maclaurin formula, meromorphic functions}

\maketitle

\vspace{-1cm}

\tableofcontents

\allowdisplaybreaks
\setcounter{section}{0}
\section*{Introduction}
 "Counting" an infinite number of points   might seem pointless and a lost cause; it has nevertheless been the concern of many a mathematician  as far back as Leonhardt Euler 
and Bernhardt Riemann
and relates to   renormalization issues in quantum field theory.

 We want to "count" lattice points on rational polyhedral convex cones, starting from the one dimensional cone $\R_+$ with lattice points  given by the positive integers studied in the first section. Evaluating the Riemann zeta function  at zero provides one way of "counting" the positive integers. It indeed assigns a finite value $\zeta(0)=-\frac{1}{2}$ to the ill-defined sum "$\sum_{n=1}^\infty n^0$"  by means of an analytic continuation $\zeta(z)$ of the regularized sum $\sum_{n=1}^\infty n^{-z}$. Alternatively the "number" $ \frac{1}{2}= 1-\frac{1}{2}$ of non-negative integers  can be derived using an alternative approximation $S(\e)  = \sum_{n=0}^\infty e^{-\e n}$ by an exponential sum. Its analytic extension (denoted by the same symbol $S$)       presents a simple pole at $\e=0$ with residue $1$ so that    $S(\e)=\frac{1}{\e}+S_+(\e)$   where $ S_+ $  is holomorphic at zero. Coincidentally, the  "polar part" $\frac{1}{\e}$  equals the integral $I(\e)= \int_{0}^\infty e^{-\e x} dx
 $   leading to the Euler-Maclaurin formula $S = I + \mu $ which relates the sum and the integral of the map $x\mapsto e^{-\e x}$ by means of the interpolator $\mu=S_+$. Using the terminology borrowed from physicists, we refer to the decomposition
    $S(\e)=\frac{1}{\e}+S_+(\e)$ into a "polar part" $\frac{1}{\e}$  and a holomorphic part $S_+(\e)$  as  the minimal subtraction scheme applied to $S$. For this particular function, it coincides with the Euler-Maclaurin formula and we have $S_+(0)=\mu(0)= \zeta(0)+1=\frac{1}{2}$.

 The coincidence in the case of the discrete exponential sum,  between the minimal subtraction scheme and the Euler-Maclaurin formula,   carries  out to higher dimensions. The second section is dedicated to   "counting" the lattice points $\Z_{\geq 0}^k$ of a (closed) product cone $\R_{\geq 0}^k$ of dimension $k\in \N$. One expects the "number" of points of $\Z_{\geq 0}^k$ to be the $k$-th power of the "number" of points of $\Z_{\geq 0}$ and this is indeed the case provided one "counts  carefully". By this we mean that one should not naively evaluate the "holomorphic part" of the $k$-th power $S^k(\e)$  at zero of the exponential sum but instead take the $k$-th power  $S_+^k(0)$ of the holomorphic part  $S_+$ evaluated at zero, which is a straightforward procedure in  the rather trivial case of product cones. However there  is a general algebraic construction which derives $S_+^k$ from $S^k$, known as the algebraic Birkhoff factorization   that can be viewed as a generalization to higher dimensions of the minimal subtraction scheme mentioned above.  It relies on a coproduct  on (product)  cones built from a complement map   described in Section 3, which separates a face of the cone from the remaining faces. When applied to the multivariable  exponential sum $\widetilde S_k:(\e_1,\cdots,\e_k)\mapsto \prod_{i=1}^kS(\e_i)$    on the product cone $\R_{\geq 0}^k$,  the general algebraic Birkhoff factorization   on  coalgebras  described in Section 4    gives   $(\e_1,\cdots,\e_k)\mapsto \prod_{i=1}^kS_+(\e_i)$  as the "renormalized holomorphic" part of the map  $\widetilde S_k$. This algebraic Birkhoff factorization   can also be interpreted as an Euler-Maclaurin formula for it factorizes the sum as a (convolution) product of  integrals   and  interpolators on product cones.

We close this presentation by briefly mentioning the corresponding result on  rational polyhedral (lattice) cones, namely that the Euler-Maclaurin formula (first derived in \cite{BV}, see also \cite{B})  for the exponential sum  is given by  its algebraic Birkhoff factorization, leaving out   the precise statement   for which we refer to reader to \cite{GPZ3}.  Renormalized  conical zeta values associated to a cone $C$ correspond to the Taylor coefficients of the "holomorphic part" $S_+(C)$ of the multivariable exponential sum $S(C)$ on the cone. For Chen cones $x_k\leq\cdots\leq x_1$ they   yield renormalized multiple zeta values; in Section  5 we illustrate our approach with the computation of renormalized multiple zeta values  with 2 and 3 arguments.
The algebraic Birkhoff factorization  on cones, seen as a general device  which renormalizes any conical zeta value  at non-positive integers,  therefore yields   a geometric approach to renormalize multiple zeta values at non-positive integers. This geometric approach  contrasts with other approaches such as  \cite{MP} and \cite{GZ}  to the renormalization of multiple zeta values at   non-positive  arguments, where the algebraic Birkhoff factorization   is carried out on the summands (functions $(x_1,\cdots, x_k)\mapsto x_1^{-s_1}\cdots x_1^{-s_k}$) rather than on the domain (the cones) of summation as in our present construction or \cite{Sa} where a purely analytic renormalization method is implemented, which does not use algebraic Birkhoff factorization.

To conclude,  "counting" lattice points on cones which might a priori seem like a very specific issue, actually brings together i) renormalization methods \`a la Connes and Kreimer \cite{CK} borrowed from quantum field theory in the form of algebraic Birkhoff factorization, ii)  the Euler-Maclaurin formula  on cones and hence on polytopes used to study the geometry of toric varieties, iii) number theory with the conical zeta values (introduced in \cite{GPZ2}) that generalize multiple zeta values~\cite{Ho,Za}, and  which arise in our context as the Taylor coefficients of the interpolator in the Euler-Maclaurin formula.  We hope that this presentation which does not claim to be neither exhaustive nor new  since it relies on previous work by the authors, will act as an incentive for the lay reader to get further acquainted with renormalization methods.

\section{Counting integers}

We want to count the non-negative integer points i.e. to evaluate the ill-defined sum "$1+1+\cdots +1+\cdots=\sum_{n=0}^\infty n^0$" and more generally  the no better defined  sum $\sum_{n=0}^\infty n^k$ for any non-negative integer $k$.

\subsection{Approximated sums over integers} We first approximate these ill-defined sums; there are at  least three ways to do so \footnote{We refer the reader to \cite{P} for a more detailed description of these various regularization methods.}:
\begin{enumerate}
\item The {\bf cut-off regularization } only considers a finite number of terms of   the sum. For $N\in \N$ we set $S_k(N):= \sum_{n=0}^N n^k$;
\item The {\bf heat-kernel type regularization }  approximates the summand by an exponential expression. For positive $\e$ we set
\begin{equation}
\label{eq:1dimsum}S(\e):=   \sum_{n=0}^\infty e^{-\e n}
\end{equation}  and $S_k(\e):=   \sum_{n=0}^\infty n^k e^{-\e n}= (-1)^k\partial_k S(\e);$
\item The {\bf zeta-function type regularization} approximates the summand by a complex power. For a complex number  $z$ whose real part is larger than $1,$ the expression
$$\tilde S(z):=   \sum_{n=1}^\infty n^{-z}=: \zeta(z) $$ called the {$\zeta$-function} converges and
 $\tilde S_k(z):=   \sum_{n=1}^\infty n^{k-z} =  \zeta(z-k)  $  converges for any complex number  $z$ whose real part is larger than $k+1$.
\end{enumerate}
 The  sums $S$ and $\tilde S$ relate via the Mellin transform; for any positive number $\lambda$ the map $f_\lambda:\e\mapsto e^{-\lambda \e} $ defines a Schwartz function whose  {\bf Mellin transform} reads $${\mathcal M}\left(f_\lambda\right)(z):=\frac{1}{\Gamma(z)}\, \int_0^\infty \e^{z-1}f_\lambda(\e)\, d\e=\lambda^{-z}.$$
\\
$\tilde S(z)= \sum_{n=1}^\infty {\mathcal M}\left(f_n\right)(z)={\mathcal M}\left(S-1\right)(z)$. This extends to an identity of meromorphic functions
  $$\tilde S = {\mathcal M}\left(S-1\right) $$
with simple poles at integers smaller or equal $1$. It turns out that the residue at $z=1$ is one and zero elsewhere.

The sum $S(\e)=\frac{1}{1-e^{-\e}}$ can be expressed in terms of the {\bf Todd function}\footnote{There are two variants of the Todd function; in topology it is  defined as the map   $\tau:\e\mapsto\frac{\e}{ 1-e^{-\e}}$, an alternative definition one finds in the literature is $\e\mapsto \tau(-\e)$, which we opt for in these notes. }.
  \begin{equation}\label{eq:Todd1}
      {\rm Td}(\e):=\frac{\e}{ e^{\e}-1}
     \end{equation}  as $$S(\e)= \frac{{\rm Td}(-\e)}{\e}.$$
The {\bf Todd function} is the  exponential generating function  for the  {\bf Bernoulli numbers}\footnote{They were discovered by Jakob Bernoulli  and independently by a Japanese mathematician Seki K\"owa, both of whose  discoveries were posthumously published (in 1712 for Seki K\"owa, in his  work Katsuyo Sampo, in 1713 for Bernoulli,  in his Ars Conjectandi).} that  correspond to the  Taylor coefficients\footnote{They also arise  in the Taylor series expansions of the tangent and hyperbolic tangent functions.}
     \begin{equation} \label{eq:Bernoullinumb}
     {\rm Td}(\e)=  \sum_{n=0}^\infty B_n \frac{\e^n}{n!}.
     \end{equation}
We have $$ {\rm Td}(\e)=\frac{\e}{ e^{\e}-1}=\frac{\e}{ \e+\frac{\e^2}{2}+ o(\e^2)}=\frac{1}{ 1+\frac{\e }{2}+ o(\e)}= 1- \frac{\e}{2}+ o(\e)$$
so $B_0=1; B_1= -\frac{1}{2}$. Since $\frac{\e}{e^\e-1}+\frac{\e}{2}= \frac{\e}{2}
\frac{e^{\frac{\e}{2}}+
     e^{-\frac{ \e}{2}}}{e^{\frac{\e}{2}}- e^{-\frac{\e}{2}}}$ is an even function,
   $B_{2k+1}=0$ for any positive integer $k$.

   Consequently, for any positive integer $K$ we have
   \begin{equation} \label{eq:ToddBernoulli}   {\rm Td}(\e)=  1- \frac{\e}{2}+\sum_{k=1}^K \frac{B_{2k }}{(2k )!}\e^{2k} + o(\e^{2K})\end{equation}
   and
      \begin{equation} \label{eq:SBernoulli}S(\e)= \frac{{\rm Td}(-\e)}{\e}= \frac{1}{\e}+ \frac{1}{2}+\sum_{k=1}^K \frac{B_{2k }}{(2k )!}\e^{2k-1} + o(\e^{2K}).\end{equation}
  \subsection{The one-dimensional Euler-Maclaurin formula }
     As a consequence of formula (\ref{eq:SBernoulli}), the discrete sum $S(\e):=\sum_{k=0}^\infty e^{-\e k}= \frac{1}{1-e^{-\e}}$ for positive $\e$  relates to the integral
 \begin{equation}
 \label{eq:1dimintegral}I(\e):=\int_{0}^\infty  e^{-\e x}\, dx= \frac{1}{\e}
  \end{equation} by means of the interpolator
    $$\mu(\e):= S(\e)- I(\e)= S(\e)-\frac{1}{\e}=\frac{1}{2} +\sum_{k=1}^K \frac{B_{2k }}{(2k )!}\e^{2k-1} + o(\e^{2K})\quad\text{ for all } K\in \N,$$
 which is holomorphic at $\e=0$.
 This interpolation formula between the sum and the integral
    \begin{equation}
    \label{eq:EulerMaclaurinexp}   S(\e)=I(\e)+\mu(\e)
    \end{equation}
 generalizes to other $L^1$ functions by means of the Euler-Maclaurin formula.

As a motivation for the Euler-Maclaurin formula with remainder, let us first derive a formal Euler-Maclaurent formula using the Todd function ${\rm Td}(D)$   obtained by inserting the derivation map $D:f\mapsto f^\prime$ on $C^\infty  (\R)$ in formula (\ref{eq:Todd1}).

Let $\nabla f(x)= f(x)- f(x-1)$ denote the discrete derivation. Using a formal Taylor expansion,  we  have $$\nabla f(x)= f(x)-f(x-1) =\sum_{k=1}^\infty \frac{(-1)^{k+1}D^k f(x)}{k!}=\left(1-e^{-D}\right)(f)(x) $$ and hence at any non-negative integer $n$
\begin{eqnarray*}
\left(\nabla^{-1}f\right)(n) &=& \left(1-e^{-D}\right)^{-1}(f)(n)\\
&=& \left(D^{-1}f\right)(n)+ \frac{1}{2} +\sum_{k=1}^K \frac{B_{2k }}{(2k )!}\left(D^{2k-1}f\right)(n) + o(\e^{2K}).
\end{eqnarray*}
Here
     $\left(\nabla^{-1}f\right)(n)=\sum_{k=0}^n f(k)+ C$    stands for the discrete  primitive of $f$ defined modulo a constant; it satisfies
    $\left(\nabla \circ \nabla^{-1}\right)f (n)= f(n)$ for any $n\in \Z_{\geq 0}$. Similarly,   $\left(D^{-1}f\right)(x)=\int_0^x f(y)\, dy+ C $ stands for the continuous integration map defined modulo a constant;  it satisfies  $\left(D\circ D^{-1} \right)f(x)= f(x)$ for any $x\in \R$. This gives a first formal expansion
      \begin{eqnarray*}  \sum_{n=a}^{b}
           f(n)&=&\left(\nabla^{-1}f\right)(b)-\left(\nabla^{-1}f\right)(a)\\
           &=& \left(D^{-1}f\right)(b)-\left(D^{-1}f\right)(a)+ \frac{f(a)+f(b)}{2}  \nonumber\\
         &+&
         \sum_{j=1}^J  \frac{B_{2j}}{(2j)!}\left( D^{2j-1}f(b)-D^{(2j-1)}f(a)\right) \\
         &=&\int_a^b f(x)\, dx + \frac{f(a)+f(b)}{2} \\
              &+&
              \sum_{j=1}^J  \frac{B_{2j}}{(2j)!}\left( f^{(2j-1)}(b)-f^{(2j-1)}(a)\right),
         \end{eqnarray*}
for any two non-negative integers $a$ and $b$.

We are now ready to state the  {\bf Euler-Maclaurin formula} with remainder \cite{Ha}

\begin{prop}\label{prop:EML}
      For any  function $ f$ in $ C^\infty(\R)$  and any two integers $a<b$,
     \begin{eqnarray}\label{eq:EulerMaclaurin} \sum_{n=a}^{b}
       f(n)&=&\frac{f(a)+f(b)}{2} +\int_a^b f(x)\, dx\nonumber\\
     &+&
     \sum_{j=1}^J  \frac{B_{2j}}{(2j)!}\left( f^{(2j-1)}(b)-f^{(2j-1)}(a)\right)\nonumber\\
     &-& \frac{1}{(2J)!} \int_a^b \overline{B_{2J}} (x)\, f^{(2J)}(x)
     \, dx
     \end{eqnarray}where $J$ is
     any positive integer and $\overline B_n(x):= B_n\left(x-\lfloor x\rfloor\right)$ built from the Bernoulli polynomials $($see e.g. \cite{A}$)$ $B_n(x):= \sum_{k=0}^n {n\choose k} \, B_{n-k} \, x^k$ and $\lfloor x\rfloor$ the integral part of $x$.
\end{prop}

In particular, for $f(x)= x^k$ and $a=0, b=N$ we have $f^{(2j-1)}(x)= \frac{k!} {(k-2j+1)!} x^{k-2j+1}$ and hence
   \begin{eqnarray}\label{eq:EMLpol}
S_k(N):&=&   \sum_{n=0}^{N} n^k\\
       &=&\frac{\delta_{k}+N^k}{2} +\int_0^N x^k\, dx+
     \sum_{j=1}^{\left[\frac{k+1}{2}\right]}  \frac{B_{2j}}{{2j}!}\left(\frac{k!}{(k-2j+1)!} \left(N^{ k-2j+1}-\delta_{k-2j+1}\right)\right)\nonumber\\
      &=&   \frac{N^{k+1}}{k+1}+ \frac{ N^k}{2}+ \sum_{j=1}^{\left[\frac{k+1}{2}\right]}{k\choose 2j-1}\frac{B_{2j}}{2j}\,
      \left(N^{ k-2j+1}-\delta_{k-2j+1}\right) + \frac{\delta_k} {2},
        \end{eqnarray}
         and we recover this way the well-known formulae
     $$\sum_{n=0}^{N}
            n^0= N+1;\quad \sum_{n=0}^{N}
                   n=\frac{N(N+1)}{2};\quad \sum_{n=0}^{N}
                                      n^2=\frac{N(N+1)(2N+1)}{6}$$
using the fact that $B_2= \frac{1}{6}$ for the last one.  More generally, it follows from Eq.~(\ref{eq:EMLpol}) that

\begin{cor} The cut-off discrete sum $S_k(N)$ is a polynomial of order $k+1$ in $N$ which vanishes at zero for any positive integer $k$.
\end{cor}

\subsection{Evaluating meromorphic functions at poles}

Let  $ {\rm Mer}^k_0(\C)$ be the set of germs of meromorphic functions  at zero
\footnote{i.e. equivalence classes of meromorphic functions
 defined on a neighborhood of zero for the equivalence relation $f\sim g $ if $f$
and $g$ coincide on some open neighborhood of zero.} with
poles at zero of order no larger than $k$, and let $$ {\rm Mer}_0(\C)=\cup_{k=0}^\infty {\rm Mer}^k_0(\C).$$ Let ${\rm
  Hol}_0(\C)$ (also denoted by ${\rm Mer}^0_0(\C)$) be the set of germs of holomorphic functions at zero.

If $f$ in ${\rm Mer}_0^k(\C)$ reads $f(z)= \sum_{i=-k}^\infty a_i z^i$, then for any $j\in\{1, \cdots, k\}$ we set
${\rm Res}_0^j(f):= a_{-j}$, called the {\bf $j$-th residue} of $f$ at zero.

The projection  map
\begin{eqnarray*}
\pi_+:{\rm Mer}_0(\C) &\to & {\rm Hol}_0(\C)\\
f&\mapsto &  \left(z\mapsto f(z)-  \sum_{j=1}^k\frac{{\rm
Res}^j_0(f)}{z^j}\right)\quad {\rm for}\quad f\in {\rm Mer}_0^k(\C)
\end{eqnarray*}
corresponds to what physicists call a {\bf minimal
subtraction scheme}. Whereas $\pi_+(f)$ corresponds to the holomorphic part of $f$,
$ \pi_-(f):= (1-\pi_+)(f) $ corresponds to the  ``polar part'' of $f$.
\begin{ex}
With the notation of the previous paragraphs, we have
\begin{equation}
\label{eq:piplusS}
S_+:=\pi_+\circ S(\e)= \mu(\e); \quad S_-:=\pi_-\circ  S(\e)= I(\e).\end{equation}
Thus the Euler-Maclaurin formula (\ref{eq:EulerMaclaurinexp}) amounts to the minimal subtraction scheme applied to $S$:
\begin{equation}
\label{eq:EML=MS} S=S_++S_-= \mu+I.
\end{equation}
An easy computation further shows that
\begin{equation}\label{eq:piplusSk}
\pi_+\circ S_k(\e)= (-1)^k\,\mu^{(k)}(\e); \quad \pi_- \circ S_k(\e)=(-1)^k\, I^{(k)}(\e).\end{equation}
\end{ex}
The holomorphic part $\pi_+(f\,g)$ of the product  of
two meromorphic functions $f$ and $g$ differs from
the product $\pi_+(f)\,\pi_+(g)$ of the holomorphic parts of $f$ and $g$ by
contributions of the poles  through  $\pi_-(f)$ and   $\pi_-(g)$ and we have
\begin{equation} \label{eq:piplus} \pi_+(f\, g)= \pi_+(f) \, \pi_+(g)+ \pi_+(f\,
  \pi_-(g))+ \pi_+(g\, \pi_-(f)).\end{equation}
    The maps $\pi_+$ and  $\pi_-$
are both  {\bf Rota-Baxter operators}  of weight $-1$ on ${\rm Mer}_0(\C)$, i.e.
$$\pi_\pm(f)\, \pi_\pm(g)= \pi_\pm(\pi_\pm(f)\,g)+ \pi_\pm(f\, \pi_\pm(g)) -\, \pi_\pm(fg).$$
We refer the reader to \cite{G} for a survey on Rota-Baxter operators.

Combining the evaluation map at zero ${\rm ev}_0:f\mapsto f(0)$ on holomorphic germs at zero with the map $\pi_+$  provides a first
regularized evaluator at zero on  $ {\rm Mer}^k_0(\C)$. The map \begin{eqnarray}\label{eq:regevzero} {\rm ev}_0^{\rm reg}: {\rm
Mer}^k_0(\C)&\to& \C\nonumber\\
 f&\mapsto & {\rm ev}_0 \circ \pi_+(f),
 \end{eqnarray}
 is  a linear form  that extends  the ordinary
evaluation map ${\rm ev}_0 $ defined on the space $
{\rm Hol}_0(\C)$.
\begin{defn}
We call a {\bf regularized evaluator} any linear extension of the evaluation map ${\rm ev}_0 $ to the space ${\rm Mer}_0(\C)$.
\end{defn}
The following result
provides a classification of  regularized evaluators.
\begin{prop}\label{thm:classregev} Regularized evaluators at zero  on ${\rm
Mer}^k_0(\C)$ are of the form:
 \begin{equation}\label{eq:regevmulti}\lambda_0=  {\rm ev}_0^{\rm reg} +
 \sum_{j=1}^k \mu_j \, {\rm Res}^j_0
\end{equation}
 for some constants $\mu_1, \cdots, \mu_k$.
In particular,  regularized evaluators at zero  on ${\rm Mer}^1_0(\C)$ are of the
form\footnote{The parameter $\mu$ that arises
here is related to  the  renormalization group parameter in quantum field theory.}
  \begin{equation}\label{eq:regevsimple}\lambda_0=  {\rm ev}_0^{\rm reg} + \mu \,
    {\rm Res}_0.
\end{equation}
 for some constant $\mu$.
\end{prop}
\begin{proof} A linear form $\lambda_0$ which extends ${\rm ev}_0$ coincides with
${\rm ev}_0$ on the range of $\pi_+$
and therefore fulfills the following identity: $$\lambda_0\circ \pi_+= {\rm ev}_0\circ
\pi_+= {\rm ev}_0^{\rm reg}.$$
Thus, for any $f\in {\rm Mer}^k_0(\C)$, using the linearity of $\lambda_0$ we get
$$\lambda_0(f)=\lambda_0\left(\pi_+(f)\right)+\lambda_0(\pi_-(f))
= {\rm ev}_0^{\rm reg}+ \sum_{j=1}^k \mu_j\, {\rm Res}^j_0(f)
$$
where we have set $\mu_j:=\lambda_0(z^{-j}) $.
\end{proof}
\begin{ex}
We have  $${\rm ev}_0^{\rm reg} (S)= \mu(0)= \frac{1}{2}= 1+\, B_1 .$$

Similarly, the higher Taylor coefficients of the holomorphic function $\mu$ at zero relate to the value of the zeta function at negative integers
$$ {\rm ev}_0^{\rm reg} (S_k)=  (-1)^k\,\mu^{(k) }(0)= -\frac{B_{k+1}}{k+1}$$
and yield the renormalized polynomial sums  $``\sum_{k=0}^\infty k^n "$  on integer points of the one dimensional closed cone $[0,+\infty)$.
\end{ex}
\subsection{The zeta function at  non-positive integers}
Let us start with some notation.
Given $\alpha \in \C$ we consider smooth functions $f$ on $\R_+$ with the following asymptotic behavior at infinity
\begin{equation}\label{eq:asympinfinity}
f(R)\sim_{R\to \infty} \sum_{j=0}^\infty  a_{j } R^{\alpha-j} +b\, \log R \end{equation}
by which we mean
$$f(R)-\sum_{j=0}^{N-1} a_{j} R^{\alpha-j}-b \log R = o\left(R^{\Re(\alpha)-N+\e}\right)$$
for any positive $\e$ and any positive integer $N$.  We call such a function {\bf asymptotically log-polyhomogeneous} at infinity of logarithmic type $1$.  If $b=0$ we call it asymptotically polyhomogeneous at infinity; let us consider the class ${\mathcal S}_\infty^{\alpha }(\R_+)$ of
{\bf asymptotically  polyhomogeneous functions} at infinity of logarithmic type $1$.

\begin{ex} The logarithmic function
$f:R\mapsto \log R$  is asymptotically log-polyhomogeneous at infinity, of logarithmic type $1$. Physicists say that the integral $\int_1^R \frac {1}{x}\,dx=\log R$ has a logarithmic divergence as $R\to \infty$.
\end{ex}
The {\bf Hadamard finite part  of $f$ at infinity}
\begin{eqnarray*} {\rm fp}_{R\to \infty} f(R)&:=& \left\{\begin{array}{ll} a_{\alpha},& \text{if }\alpha\in \Z_{\geq 0},\\ 0, &\text{otherwise}.\end{array}\right.
\end{eqnarray*}
defines a linear map
\begin{eqnarray*}{\rm ev}_\infty^{\rm reg}: {\mathcal S}_\infty^{\alpha}(\R_+)&\longrightarrow &\C\\
f&\longmapsto & {\rm fp}_{R\to \infty} f(R)
\end{eqnarray*}
which  extends the ordinary limit at infinity whenever it exists. We call such a linear  extension of the ordinary limit  a  {\bf regularized evaluator at infinity.}

 Setting $R=\frac{1}{r}$ in Eq.~(\ref{eq:asympinfinity})    with $r>0$, and choosing $\beta=-\alpha$, $b_{ j}= a_{ j}$,  $c=-b$    leads to
  smooth functions $f$ on $(0,+\infty)$
 with the  following log-polyhomogeneous asymptotic behavior at zero:

  \begin{equation}\label{eq:asympzero}
f(r)\sim_{r\to 0} \sum_{j=0}^\infty  b_{j } r^{\beta+j} +c\, \log r \end{equation}
by which we mean
$$f(r)-\sum_{j=0}^{N-1} b_{j} r^{\beta+j}-c \log r = o\left(r^{\Re(\beta)+N+\e}\right)$$
for any positive $\e$ and any positive integer $N$. We call such a function {\bf asymptotically log-polyhomogeneous} at zero of logarithmic type $1$.  If $c=0$ we call it asymptotically polyhomogeneous at zero; let us consider the class ${\mathcal S}_0^{\beta}(\R_+)$ of
{\bf asymptotically  polyhomogeneous}  functions at zero of logarithmic type $1$.

The {\bf Hadamard finite part  of $f$ at zero}
\begin{eqnarray*}{\rm fp}_{r\to 0} f(r)&:=&
\left\{\begin{array}{ll} b_{-\beta }, &\text{if } \beta\in \Z_{\geq 0},\\
0, &\text{otherwise}.\end{array}\right.
\end{eqnarray*}
defines a linear map
\begin{eqnarray*}{\rm ev}_0^{\rm reg}: {\mathcal S}_0^{\beta}(\R_+)&\longrightarrow &\C\\
f&\longmapsto & {\rm fp}_{r\to 0} f(r)
\end{eqnarray*}
which  extends the ordinary limit at zero  whenever it exists. We call such a linear  extension of the ordinary limit  a  {\bf regularized evaluator at zero.}

 We recall here well known results on the Mellin transform  \footnote{Note that definitions of the Mellin transform   differ according to the  reference  by a  multiplicative factor $\Gamma(z)$.},  see e.g. \cite{Je}.
\begin{prop}\label{prop:MellinS}Let $f$ be a  Schwartz function in ${\mathcal S}_0^\beta(\R_+)$ for some $\beta\in \C$. Its Mellin transform defines a holomorphic function on the half plane $\Re(z)+ \beta >0$ which extends to a meromorphic function on the  whole complex plane
which is holomorphic at zero. We have  ${\mathcal M}(f^{(k)})(z)= (-1)^k\,{\mathcal M}(f) (z-k)$ for any $k\in \Z_{\geq 0}$ and the value at zero is given by
\begin{equation}\label{eq:MellinHK}{\rm ev}_0^{\rm reg}\circ {\mathcal M}(f)= {\mathcal M}(f)(0)   ={\rm ev}_0^{\rm reg}(f)´.
\end{equation}
\end{prop}

\begin{proof}    We split the Mellin transform
$${\mathcal M}(f)(z)= \frac{1}{\Gamma(z)}\left( \int_0^A \e^{z-1} f(\e)\, d\e+ \int_A^\infty \e^{z-1} f(\e)\, d\e\right)$$
for some positive real number $A$. The function $ \frac{1}{\Gamma}$ is holomorphic at zero and  we have $ \frac{1}{\Gamma (z)}\sim_0 z$.

Since  $f$ is a Schwartz function, the second term   in the bracket   yields a holomorphic function $ I_2: z\mapsto \frac{1}{\Gamma(z)}\int_A^\infty \e^{z-1} f(\e)\, d\e$  which vanishes at zero. For  $f(\e)= \sum_{j=0}^J   b_{j } \e^{\beta+j} $ and  $\Re(z)+\beta+j>0 $, the first term in the bracket gives rise to $$I_1^J(z):= \frac{1}{\Gamma(z)}\, \sum_{j=0}^J   b_{j } \int_0^A \e^{z+\beta+j-1} \, d\e= \frac{1}{\Gamma(z)}\,\sum_{j=0}^J   b_{j } \frac{A^{z+\beta+j}}{z+\beta +j}$$
which extends to a meromorphic function  denoted by the same symbol.  Hence $$ {\mathcal M}(f)(z):=I_1(z) + I_2^J(z)+ o\left(\e^{z+\beta +J}\right) \text{ for } J\in \N $$
defines a meromorphic function on the whole plane. Integrating by parts $k$ times  and implementing the property $\Gamma(z)= (z-1)\Gamma(z-1)=(z-1)\cdots (z-k)\Gamma(z-k)$ gives ${\mathcal M}(f^{(k)})(z)=  (-1)^k\,{\mathcal M}(f) (z-k)$.

Since $\Gamma(z)\sim \frac{1}{z}$, the  value of $ I_2^J(z)$ at  $z=0$ is   $b_{-\beta}$  if $\beta \in \Z_{\leq 0}\cap [0, J ]$ and zero elsewhere so the same holds for    $ {\mathcal M}(f) (z) $.  Since ${\rm fp}_{\e=0}  f(\e)=b_{-\beta}$  if $\beta \in \Z_{\leq 0}$ and zero elsewhere, this yields Eq.~(\ref{eq:MellinHK}). \end{proof}

The Mellin transform of the Schwartz function $f_n:\e\mapsto e^{-\e n}$ on $\R_+$ reads $n^{-z}= {\mathcal M}(f_n)(z)$ for  any  $n \in \N$.

 \begin{cor} The function $z\mapsto\sum\limits_{n=1}^\infty n^{-z}$ defined on the half-plane $\Re(z)>1$  extends  meromorphically on  the whole plane to the {\bf zeta function} $\zeta $, which has only one simple pole at $-1$ and its value at non-positive integers is expressed in terms of Bernoulli numbers
\begin{equation}
\label{eq:zetavalues}\zeta(0)= {\rm ev}_0^{\rm reg} (  S)-1=-\frac{1}{2}  ;\quad \zeta(-k)=  (-1)^k \,{\rm ev}_0^{\rm reg}( \partial^k  S  )=-\frac{B_{k+1}}{k+1} \text{ for all } k\in \N.
\end{equation}
 \end{cor}
\begin{proof}This follows from applying Proposition \ref{prop:MellinS} to the function $\e\mapsto S(\e)-1=\sum_{n=1}^\infty e^{-\e n}$ and its derivatives $(-1)^k S_k, k\in \N$.
 \end{proof}

\subsection{Conclusion} By means of  the heat-kernel regularization method we evaluated  $$"\left(\sum_{n=1}^\infty n^0\right)"="\left(\sum_{n=0}^\infty n^0\right)"-1={\rm ev}_{0}^{\rm reg}\circ   S-1 =\mu(0)-1=B_1-1=-\frac{1}{2}. $$  In this paper, $B_1= \frac 12$.
By means of  the zeta-function regularization method we evaluated  $$"\left(\sum_{n=1}^\infty n^0\right)"={\rm ev}_0  \circ   S=\zeta(0)=-\frac{1}{2} ,$$
so these two methods agrees in the case $k=0$. The two methods actually coincide for any $k\in \Z_{\geq 0}$.
$$"\left(\sum_{n=0}^\infty n^k\right)"="\left(\sum_{n=1}^\infty n^k\right)"=  {\rm ev}_{0}^{\rm reg}\circ    S_k =\zeta(-k)  = -\frac{B_{k+1}}{k+1}.$$ Moreover, combining  Eqs.~(\ref{eq:zetavalues}) and (\ref{eq:SBernoulli}) yields the Laurent expansion of the exponential sum in terms of $\zeta$-values at non-positive arguments
\begin{equation}
\label{eq:expsumzeta} S(\e)=  \frac{1}{\e}-\zeta(0)-\sum_{k=1}^K \zeta(-(2k-1))\,\e^{2k-1} + o(\e^{2K})\quad\forall K\in \N.
\end{equation}
 In contrast, the cut-off method gives $${\rm fp}_{N\to \infty}  S_k(N)=P_k(0)= \delta_{k},$$
where $\delta_k=1$ if $k=0$ and zero otherwise.

\section{Counting lattice points on product cones}
Given a positive integer $k$, we now  want to "count" the number  $"\left(\sum_{\vec n\in \Z_{\geq 0}^k} \vec n^{\vec 0}\right)"$ of lattice  points $\vec n\in\Z_{\geq 0}^k$ in the product cone    $   \R_{\geq 0}^k $, where for $\vec n=(n_1,\cdots, n_k)\in \Z_{\geq 0}^k$ and   $\vec r=(r_1,\cdots, r_k)\in \Z_{\geq 0}^k$ we have set $\vec n^{\vec r}= n_1^{r_1}\cdots n_k^{r_k}$. We first describe the algebra of product cones.

\subsection{The  exponential summation and integration map on   product cones}

 Given a basis ${\mathcal B}_n=(e_1,\cdots, e_n )$  of  $\R^n$, let ${\mathcal P}_{{\mathcal B}_n}(\R^n)$ be the set of  {\bf product cones} $$\langle e_I\rangle  :=  \sum_{i\in I}  \R_{\ge 0}   e_i,\quad  I\subseteq [n]:=\{1,\cdots, n\},   $$
viewed  as subsets of $\R^n$. Extending this basis to a basis  ${\mathcal B}_{n+1}=(e_1,\cdots, e_{n+1} )$  of  $\R^{n+1}$, a product cone in $\R^n$ can be viewed as a  product cone in $\R^{n+1}$. Setting ${\mathcal P}_{{\mathcal B}_0}(\R^0)=\{0\}$, we define the  set
$${\mathcal P}_{\mathcal B}(\R^\infty):= \cup_{n=0}^\infty  {\mathcal P}_{{\mathcal B}_n}(\R^n)$$
of  product cones in $\R^\infty$ equipped with a basis ${\mathcal B}=\{e_n\,|\, n\in \N\}$.
Equivalently,
$$ {\mathcal P}_{\mathcal B} (\R^\infty)=\left \{ \cone{e_I}\,|\, I\subset\N \text{ finite }\right\} \text{ with } \cone{e_\emptyset}:=\{0\}.$$
It is $\Z_{\geq 0}$-filtered by the dimension ${\rm card}(I)$ (here card stands for cardinal) of the cone $\langle e_I\rangle$ and it is equipped with a partial product
  $$\langle e_I\rangle \bullet\langle e_J\rangle  :=\langle e_{I\cup J}\rangle  $$
  for two disjoint subsets $I, J $ of $\N$.
   This product is compatible with the  filtration   since the dimension of the product of two cones is the sum of their dimensions.

Unless otherwise specified, we take ${\mathcal B}$ to be the canonical basis of $\R^\infty$, in which  case we drop the subscript ${\mathcal B}$ in the notation.

The linear map  ${\rm ev}_0^{\rm reg}: {\rm Mer}_0^1(\C)\to \C$ extends multiplicatively to the subspace
${\rm Mer}_{\rm sep}(\C^\infty)$ of ${\rm Mer}_0(\C^\infty)$ spanned by separable functions
 \footnote{ ${\rm Mer}_{\rm sep}(\C^\infty)$   is isomorphic to the filtered vector space $\mathcal{F}:=\dirlim({\rm Mer}_0^1(\C))^n$ by assigning $f_1\ot \cdots \otimes f_n$ to $f_1(\e_1)\cdots f_n(e_n)$.
But the map does not respect the tensor product. For example, $f\ot f(\e_1,\e_e)=f(\e_1)f(\e_2)\neq f(\e_1)^2$. }:
$$ 
\left\{\left .f=\prod_{i\in I} f_i\,\right|\, I\subseteq \N \, {\rm finite},\, f_i\in {\rm Mer}_0(\C e_i) \right\}$$
by
\begin{equation}\mlabel{eq:evren}
{\rm ev}_{\vec 0}^{\rm ren}\left(\prod_{i\in I} f_i\right):=\prod_{i\in I} {\rm ev}_0^{\rm reg}(f_i).
\end{equation}

We refer the reader to \cite{GPZ1} for a more detailed study of renormalized (or generalized)   evaluators.
Note that even though the subspace ${\rm Mer}_{\rm sep}(\C^\infty)$ is closed under the multiplication of ${\rm Mer}_0(\C^\infty)$, the map ${\rm ev}_{\vec 0}^{\rm reg}$ on ${\rm Mer}_{\rm sep}(\C^\infty)$ resulting from Eq.~(\mref{eq:evren}) is multiplicative only for a product with disjoint variables. More precisely, for a  separable function  $f$ (resp. $g$) with variables in a finite subset $I$ (resp. $J$) of $\N$, with $I$ and $J$ disjoint, we have
$${\rm ev}_{\vec 0}^{\rm reg}(fg)={\rm ev}_{\vec 0}^{\rm reg}(f) {\rm ev}_{\vec 0}^{\rm reg}(g).$$
However, ${\rm ev}_0^{\rm reg}(\frac{1}{\e_1})\neq {\rm ev}_0^{\rm reg}(\frac{1+\e_1}{\e_1}){\rm ev}_0^{\rm reg}(\frac{1}{1+\e_1})$ even though $\frac 1{\e_1}=\frac {1+\e_1}{\e_1} \frac 1{1+\e_1}$.

The summation  map (\ref{eq:1dimsum}) and the integration map  (\ref{eq:1dimintegral}), which lie in  the linear space $ {\rm Mer}_0^1(\C)$   of meromorphic germs   in one complex variable with  a simple pole at zero, induce linear maps
on the linear space $\R\calp_\calb(\R^\infty)$ spanned by $\calp_\calb(\R^\infty)$  as follows
 $$\cals: \R\calp_\calb(\R^\infty) \longrightarrow \text{Mer}_{\rm sep} (\C^\infty), \quad \langle e_I\rangle\mapsto \prod_{i\in I}S(\e_i)$$
               and
 $${\mathcal I}: \R\calp_\calb(\R^\infty) \longrightarrow  \text{Mer}_{\rm sep} (\C^\infty), \quad \langle e_I\rangle\mapsto \prod_{i\in I}I(\e_i).$$
For simplicity and emphasizing   the dependence on the variables, we also use the notations
\begin{equation}
\cals_i: \R \cone{e_i} \to {\rm Mer}_0(\C \e_i), \quad \cone{e_i}\mapsto S(\e_i)
\mlabel{eq:si}
\end{equation}
and
\begin{equation}
\cali_i: \R\cone{e_i} \to {\rm Mer}_0(\C \e_i),\quad \cone{e_i} \mapsto I(\e_i).
\mlabel{eq:ii}
\end{equation}

The maps $\cals$ and $\cali$ are compatible with
the partial product on cones. Indeed, for two disjoint finite sets $I$ and $J$ of $\N$ we have
$${\mathcal S}\left(\langle e_{I\cup J}\rangle \right)= \prod_{i\in I \cup J}S(\e_i)= \left(\prod_{i\in I }S(\e_i)\right)
 \left(\prod_{j\in J }S(\e_j)\right)$$
  and similarly for the integration map.
 We further set ${\mathcal S}(\{0\})={\mathcal I}(\{0\})=1$.
Set $I=\{i_1,\cdots, i_J\}$ and $\e_I=(\e_{i_1},\cdots, \e_{i_J})$. As a consequence of Eq.~(\ref{eq:expsumzeta}) we have the following iterated Laurent expansion
\begin{equation}
\label{eq:expsumproductzeta} {\mathcal S}\left(\langle e_I\rangle \right)(\e_I)= \prod_{j=1 }^J \left(\frac{1}{\e_{i_j}}-\zeta(0)-\sum_{k_j=1}^{K_{j}} \zeta(-(2k_j-1))\,\e_{i_j}^{2k_j-1} + o(\e_{i_j}^{2K_j})\right)
\end{equation}

In order to "count" the number of lattice  points $"\left(\sum\limits_{\vec n\in \sum_{i\in I}\Z_{\geq 0}e_i} \vec n^{\vec 0}\right)"$ in the product cone    $  \langle e_I\rangle  $ we want to evaluate $  {\mathcal S} \left(\langle e_I\rangle \right) $ at $(\e_{i_1},\cdots, \e_{i_k})=\vec 0$.  Since  $  {\mathcal S} \left(\langle e_I\rangle \right)\in {\rm Mer}_{\rm sep}(\C^\infty)$
 a first guess is to assign the value
 \begin{equation}
 \label{eq:firstguess}
 {\mathcal  S}_{\vec 0}^{\rm ren}:={\rm ev}_{\vec 0}^{\rm ren} \circ  {\mathcal S},\end{equation}
 where ${\rm ev}^{\rm ren}_{\vec 0}$ is defined in (\ref{eq:evren}).
 This "renormalized value" at zero is multiplicative  as a result of the multiplicativity of ${\rm ev}_{\vec 0}^{\rm reg}$.
 Indeed, given two disjoint index sets $I$ and $J$, we have
  \begin{eqnarray*} {\mathcal  S}_{\vec 0}^{\rm ren}\left(\langle e_I\rangle \bullet \langle e_J \rangle \right)
  &=& {\rm ev}_{\vec 0}^{\rm reg}\circ \left( {\mathcal S} \left(\langle e_{I }\rangle \right) {\mathcal S}\left(\langle e_{ J}\rangle \right)\right)\\
 &=& \left({\rm ev}_{\vec 0}^{\rm reg} \circ \left( {\mathcal S}\left(\langle e_{I }\rangle \right)\right)\right) \cdot \left( {\rm ev}_{\vec 0}^{\rm reg}\circ  \left( {\mathcal S} \left(\langle e_{ J}\rangle\right)\right)\right) \\
    &=& {\mathcal  S}_{\vec 0}^{\rm ren}\left(\langle e_I\rangle\right) \cdot {\mathcal  S}_{\vec 0}^{\rm ren}\left(\langle e_J\rangle\right).
  \end{eqnarray*}

We shall now describe the underlying algebraic framework, which might seem   somewhat artificial  in the rather trivial product cone situation. However, on the one hand even in this simple situation it  is useful to control  the "polar part" which one needs to extract in order to define the finite part, on the other hand it offers a good  toy model to motivate otherwise relatively sophisticated techniques  which   can  be generalized  beyond   product cones, namely to general convex cones \cite{GPZ3}.

\subsection{A complement map on product cones}  Let us first recall the properties of the set complement map.

Let ${\mathcal P}_f(E)$ be the set of finite subsets of a given set $E$   equipped with the inclusion $\subseteq$ which defines a partial order compatible with the filtration of ${\mathcal P}_f(E)$ by the cardinal in the sense that $J\subseteq I$ implies $|J|\leq |I|$. For $I\in {\mathcal P}_f(E)$ let
$$\mathfrak{s}(I):=\{J\in \mathcal{P}_f(E)\,|\, J\subseteq I\} $$
be the set of subsets of $I$.

 The {\bf set complement map} assigns to any $I\subseteq E$ a map
$$ \begin{array}{rcl}
\complement_I: \mathfrak{s}(I)& \longrightarrow & \mathfrak{s}(I)\\
& &\\
J&\longmapsto& I\setminus J:= I\cap \overline J.
\end{array}$$
The complement $I\setminus J$ satisfies the following properties:
\begin{enumerate}
\item  {\bf Compatibility   with the partial order:} Let $I,J \in \mathcal{P}_f(E)$ be such that $ J\subseteq I$. Then
$$\text{for any } H\in {\mathcal P }_f(E) \text{ with } H\subseteq I\setminus J \text{ there exists unique } K\in  \mathcal{P}_f(E)  ;\, J\subseteq K\subseteq  I \text{ such that } H=I\setminus K.$$
\item  {\bf Transitivity:} Let $  I,J,K \in \mathcal{P}_f(E)$ be such that $K\subseteq J\subseteq I$. Then
    $$\left(I\setminus K\right)\setminus\left(J\setminus K\right)=I\setminus J$$
\item {\bf Compatibility   with the filtration:}  Let $I,J \in \mathcal{P}_f(E)$ be such that $ J\subseteq I$. Then
$${\rm card} (J)+ {\rm card} (I\setminus J)= {\rm card} (I ),$$
where card stands for the cardinality.
\end{enumerate}

The set complement map on $\ZZ _{\ge 0}$ induces a complement map on  the product cones. Let us first introduce some notations.   Faces of the product cone $ C:=\langle e_I\rangle $ are of the form $$F_J :=\langle e_J\rangle $$ with $J\subset I$, each of them defining  a cone with  faces  $F_{J^\prime}$ where $  J^\prime\subset J$.   The cone  $C$ therefore has  $2^{\vert I\vert}$ faces, as many as subsets of $I$.  The  set ${\mathcal F}\left(C\right) $   of faces of the cone $ C$  is equipped with a partial order
$$F^\prime\subset F \text{ if and only  if } F^\prime \text{ is a face of } F^\prime$$
or equivalently, $F_{J^\prime}\subset F_J$ if and only if $J^\prime\subset J.$  For    $F^\prime=F_{J^\prime}\subset  F=F_J$ we consider   the complement set $\overline{F^\prime}^F:=  F_{J\setminus J^\prime}$, which again defines an element of ${\mathcal F}\left(C\right)$ and hence a cone.  We   define the complement map
\begin{eqnarray}\label{eq:complementproductcone} {\mathcal F}\left(C\right)&\longrightarrow & {\mathcal F}\left(C\right)\nonumber\\
F_J&\longmapsto& {\overline F_J }^{C} =  F_{I\setminus  J },
\end{eqnarray}
which is an involution. As a consequence of the properties of the set complement map, it enjoys the following properties. Let $F\in {\mathcal F}\left(C\right)$.
\begin{enumerate}
\item {\bf  Compatibility with the partial order: } There is a one-to-one correspondence between the set of faces of  $C$ containing a given face $F$ and the set of faces of the cone ${\overline  F}^{C}$; for any face $H$ of ${\overline  F}^{C}$, there is a unique face $G$ of $C$ containing $F$ such that $H= \overline F^G$.
\mlabel{it:com10}
\item {\bf Transitivity: } ${\overline  F}^{C}= \left(\overline{   F^\prime}^{F} \right)^{\overline { F^\prime}^{C}}$ if $F^\prime\subset F $.
\mlabel{it:tra0}
\item   {\bf Compatibility with the filtration by the dimension:}  For any face $F$ of $C$ we have ${\rm dim}(F) + {\rm dim}\left({\overline  F}^{C}\right)={\rm dim }(C) $.
\end{enumerate}

There is an alternative description of this complement map which is  generalizable to general convex cones, those not necessarily obtained as product cones.
For this we observe that for a face
$   F=F_J=\langle e_J\rangle $ of a product cone $C= \langle e_I\rangle$, we have \begin{equation}
\label{eq:perpcomplement}\overline{F }^C=  F_{I\setminus J }=\pi_{F^\perp}(C),
\end{equation}
where $F^\perp$ denotes the orthogonal space of the linear space spanned by the cone $F$ in the linear space  $\langle C\rangle $  spanned by $C$, and
$\pi_{F^\perp}$ is the orthogonal projection from $\langle C\rangle$ onto $F^\perp$. Here the orthogonal projection is taken with respect to the canonical Euclidean product on $\R^\infty$. Eq.~(\ref{eq:perpcomplement}) follows from the fact that $\pi_{F^\perp}(e_i)$ is $0$ for $i\in J$ and $e_i$ for $i\notin J$.

\subsection{Algebraic Birkhoff factorization on product cones}

For each $i\geq 1$, the algebra $\cala_i:={\rm Mer}_0(\C\e_i)$ is naturally isomorphic to $\cala:={\rm Mer}_0(\C)$ as the algebra of Laurent series. Following the minimal subtraction scheme we have a direct sum ${\mathcal A}_i={\mathcal A}_{i,+}\oplus  {\mathcal A}_{i,-}$ of two subalgebras ${\mathcal A}_{i,\pm}:= \pi_\pm\left({\mathcal A}_i\right)$. The maps ${\mathcal S}_i:\R\langle e_i\rangle\longrightarrow {\rm Mer}_0(\C\e_i)$ defined in Eq.~(\mref{eq:si}) split   accordingly ${\mathcal S}_i={\mathcal S}_{i,+}+ {\mathcal S}_{i,-}$
into a sum of maps ${\mathcal S}_{i,\pm}: \R\langle e_i\rangle \longrightarrow {\mathcal A}_{i,\pm}$.

We next consider separable functions in several variables. For disjoint subsets $I, J\subseteq \N$, define
$$ \cala_{I,+,J,-}:=\left(\prod_{i\in I}\cala_{i,+}\right)\,\left(\prod_{j\in J}\cala_{j,-}\right).$$
Also denote $\cala_{I,+,J,-}=\cala_{I,+}$ if $J=\emptyset$.
Then we have
$$ \cala_I:=\prod_{i\in I} \cala_i= \oplus_{I_1\sqcup I_2=I}\cala_{I_1,+,I_2,-}\,.$$
Further denote
$$ \cala_{I,+}:=\prod_{i\in I}\cala_{i,+}, \quad
\cala_{I,-}:=\prod_{J\subsetneq I} \cala_{J,+,I\backslash J,-}$$
and
$$\cala_\infty:=\dirlim \cala_I, \quad \cala_{\infty, \pm}:=\dirlim \cala_{I,\pm}.$$
Then we have
$$ \cala_\infty=\cala_{\infty,+}\oplus \cala_{\infty, -}.$$
$\cala_{\infty,+}$ is a subalgebra but not $\cala_{\infty,-}$. For example, $\cala_{1,+}\cala_{2,-}$ and $\cala_{1,-}$ are in $\cala_{\infty,-}$, but their product is not.
This should give the decomposition for us to use. It is the restriction of the decomposition on meromorphic functions with linear poles given in~\mcite{GPZ4}.

As we saw in the previous section, since $  S_-=I$, such a splitting   $S=S_++ S_-$ of the exponential sum   corresponds to  the Euler-Maclaurin formula $S= \mu+ I$ with $\mu= S_+$.

We are  now ready to generalize the minimal subtraction scheme and the Euler-Maclaurin formula to product cones. In the product cone framework, the minimal subtraction scheme generalizes  to an elementary form of  the more general algebraic Birkhoff factorization    on coalgebras which we shall describe in the next section.

\begin{prop} Given a product cone $C=\langle e_I\rangle$ in  ${\mathcal P}(\R^\infty)$ the map $ {\mathcal  S}\left(C\right)$  extends to a  meromorphic map   in  $ {\rm Mer}_{\rm sep}(\C^\infty)$   with simple poles on the intersections of hyperplanes $\cap_{j\in J}\{\e_j=0\}$ corresponding to faces $F_J=\langle e_I\rangle, J\subseteq I$ of the cone $C$.  It   decomposes as
\begin{eqnarray} {\mathcal  S}\left(C\right)
&=&  \sum_{F\in {\mathcal F}\left(C\right)}    {\mathcal  S}_+({\overline F^C}) {\mathcal  S}_-(F)    \quad \text{\rm(algebraic Birkhoff factorization)}\label{eq:prodBHF} \\
&=&  \sum_{F\in {\mathcal F}\left(C\right)  }  \mu(\overline F^C)  {\mathcal  I}(F)  \quad\text{\rm(Euler-Maclaurin formula)} \label{eq:prodEML},
\end{eqnarray}
where for a face $F = \langle e_K\rangle $ of the cone $ C$, $\overline F^C=F_{I\setminus K}$ is the ``complement  face" defined in the previous paragraph and where  we have set
$${\mathcal S}(F)=  {\mathcal S} \left(\langle e_K\rangle \right):=\prod_{i\in K}S(\cone{e_i}),\quad  {\mathcal  S}_\pm (F) =   {\mathcal S}_{\pm}\left(\langle e_K\rangle \right):=\prod_{i\in K}\cals_{i,\pm}.$$
\mlabel{pp:sabf}
\end{prop}
\begin{remark}   Eq.~(\ref{eq:prodBHF}) which arises from the one-dimensional minimal substraction scheme  can be viewed as a higher dimensional minimal subtraction scheme and Eq.~(\ref{eq:prodEML})   as a higher dimensional  Euler-Maclaurin formula.  When $k=1$ they yield back the one dimensional  minimal subtraction scheme and the Euler-Maclaurin formula applied to $S(\e)$.
\end{remark}
\begin{proof} Let $C  =\langle e_I\rangle$ for some finite subset $I$ in $\N$. We have
\begin{eqnarray*} {\mathcal  S}\left(C\right) &=&\prod_{i\in I}  {\mathcal  S}_i\left(\langle e_i\rangle \right) \quad\text{(a product of sums)}\nonumber\\
& = &\prod_{i\in I} \left(  {\mathcal  S}_{i,+}+ {\mathcal  S}_{i,-}\right)\left(\langle e_i\rangle \right)\quad \text{(a sum of products)}\label{eq:productsums}\\
&=&  \sum_{J\subset I}\left( \prod_{j\in I \setminus J} {\mathcal  S}_{j,+}\left(\langle e_j\rangle \right)\right)\,\,\left(\prod_{j\in J}{\mathcal  S}_{j,-}\left(\langle e_j\rangle \right)\right)\label{eq:sumproducts}\\
&=&  \sum_{F\in {\mathcal F}\left(C\right)} {\mathcal  S}_+({\overline F^C}) \, {\mathcal  S}_-(F)   \quad\label{eq:prodBHFcone} \\
&=&  \sum_{F\in {\mathcal F}\left(C\right)  } \mu(\overline F^C) \,  {\mathcal  I}(F). \quad \label{eq:prodEMLcone}
\end{eqnarray*}
\end{proof}

The fact that the algebraic Birkhoff factorization   (\ref{eq:prodBHF})  and the Euler-Maclaurin formula  (\ref{eq:prodEML})  coincide for product cones is a consequence of Eq.~(\ref{eq:EML=MS}) which shows how,  in the one dimensional case,  the minimal subtraction scheme and the Euler-Maclaurin formula   coincide for the exponential sum.
From formula (\ref{eq:expsumproductzeta}) we  derive a Taylor expansion at zero of $ {\mathcal S}_+\left(\langle e_I\rangle \right)$
\begin{equation}
\label{eq:expsumplusproductzeta} {\mathcal S}_+\left(\langle e_I\rangle \right)(\e_I)= \sum_{J\subset I }  \left( \zeta(-k_J)\,\e_{J}^{k_J} + o(\e_{J}^{k_J})\right)
\end{equation}
where for $J=\{i_1,\cdots, i_j\}\subset I$ and any multiindex $k_J=(k_{i_1},\cdots,k_{i_j} )\in \Z_{\geq 0}^j$  we have set $\e_J^{k_J}= \prod_{l=1}^j\e_{i_l}$, whose coefficients
\begin{equation}
\label{eq:productmutltizeta}
\prod_{j\in J}\zeta(-k_{i_j})\end{equation}
are the so called {\bf (renormalized) product zeta values} at non-positive integers.

The renormalized discrete sum in Eq.~(\ref{eq:firstguess}),  which can   be rewritten as
$$   {\mathcal  S}_{\vec 0}^{\rm reg}={\rm ev}_{\vec 0}  \circ  {\mathcal S}_+= {\rm ev}_{\vec 0}  \circ \mu, $$
is obtained from evaluating at zero the  renormalized ``holomorphic part" ${\mathcal S}_+$ of the exponential sum
derived from the algebraic Birkhoff factorization  (see (\ref{eq:prodBHF}))   or equivalently  from evaluating at zero    the   renormalized interpolator $\mu$  derived from the Euler-Maclaurin formula \ (see (\ref{eq:prodEML})).

We have gone a long way around to recover our first guess (\ref{eq:firstguess}). This approach using Birkhoff-Hopf factorization,
even if somewhat artificial in the case of product cones, is nevertheless  useful  for it can be generalized to all rational polyhedral convex (lattice) cones
\cite{GPZ3} a case which will be briefly discussed at the end of the paper.

 \section{From complement maps to  coproducts} We now set up an algebraic framework  to derive an  algebraic Birkhoff factorization   from a complement map in a more general set up than the   specific example of product cones   which served as a toy model in the previous section.

 \subsection{Posets} Let $(\mathcal{P}, \leq)$ be a poset, i.e. a   set $\mathcal{P}$ together with a partial order $\leq $. We do not assume that the poset is finite.

 The poset is  {\bf filtered} if $ \mathcal{P}=\bigcup_{n=0}^\infty  \mathcal{P}_n$ with $ \mathcal{P}_n\subset  \mathcal{P}_{n+1}$.  The degree of $A\in \mathcal{P}$ denoted by $|A|$ is  the smallest integer $n$ such that $A\in \mathcal{P}_{n}$. The partial order $\leq$ is compatible with the filtration if $A\leq B$ implies $\vert A\vert \leq \vert B\vert.$

We call  a  filtered poset $\mathcal{P}$ {\bf connected} if $\calp$ has a least element $1$, called the {\bf bottom} of $\calp$, and we have ${\mathcal P}_0=\{1\}$.

\begin{ex}
For a given set $X$ (finite or infinite), the set $\mathcal{P}_f(X) $    of finite subsets of $X$ equipped with the inclusion relation is a    poset $(\mathcal{P}_f(X), \subseteq)$  filtered by the cardinal.
It is connected since $\emptyset$ is the only subset of cardinal $0$ and $\emptyset\subseteq A$ for any $A\in \mathcal{P}_f(X) $.
\end{ex}

\begin{ex}
This example can be regarded as a special case of the previous example but  its  pertinence for convex cones  justifies that we treat it separately. The set $\mathcal{P}(\R^\infty)=\cup_{n=0}^\infty \mathcal{P}(\R^n)$  of  closed product cones  described in the previous section  is filtered by the dimension and  partially ordered  by the partial order on the index sets. Equivalently,   $F\leq C$ if the product cone $F$ is a face of  the  product cone $C$. $\mathcal{P}(\R^\infty)$ is connected since  the zero cone $\{0\}$  is the only cone of dimension $0$ and $\{0\}\leq C$ for any $C\in \mathcal{P}(\R^\infty)$ as $0$ is a vertex of any product  cone.
\end{ex}

\begin{ex}
A closed (polyhedral) convex cone  in $\R^n$ is the convex set
\begin{equation}
\langle v_1,\cdots,v_k\rangle :=\R _{\geq 0}v_1+\cdots+\R_{\geq 0}v_k,
\mlabel{eq:cone}
\end{equation}
where $v_i\in \R^n$, $i=1,\cdots, k$.

Let $\mathcal{C}(\R^\infty)=\cup_{n=0}^\infty \mathcal{C}(\R^n) $  be the set of closed polyhedral convex cones    in $ \R^\infty $   see \cite{GPZ2}. We have $ \mathcal{P}(\R^\infty)\subset \mathcal{C}(\R^\infty)$. It is  filtered by the dimension $\vert C\vert$ of the cone  $C $  defined as the dimension of the linear subspace spanned by $C$. A  face  of a cone $C=\langle v_1,\cdots,v_k\rangle$   is a subset of the form $\cone{v_1,\cdots,v_k}\cap \{u=0\}$, where $u:\R^n\to \R$ is a linear form which is non-negative on $\langle v_1,\cdots,v_k\rangle$. A face $F$ of a cone is itself  a cone and we equip $\mathcal{C}(\R^\infty)$   with the following partial order which extends the partial order on product cones:
$$F\leq C \text{ if and only if } F\text{ is a face of C},$$
which   is compatible with the filtration since
$ F\leq  C$ implies ${\rm lin}(F)\subset {\rm lin}(C)$ which implies $ \vert F\vert \leq \vert C\vert.$
The filtered poset $\left(\mathcal{C}(\R^\infty),\leq \right)$ is connected since the zero cone $\{0\}$  is the only cone of dimension $0$ and $\{0\}\leq C$ for any $C\in \mathcal{C}(\R^\infty)$ since $0$ is a vertex of any cone pointed at zero.
\end{ex}
\begin{ex} \label{ex:tree}
A {\bf planar rooted tree} (see e.g. \cite{CK,F,M})  is a finite connected directed graph, without cycles, together with an embedding of it into the plane, such that only one vertex (the root) has outgoing edges only. We consider the   set  $\mathcal{T} $ of   planar rooted trees filtered by the number of vertices. Concatenations of trees give rise to forests.

An elementary cut on a tree is a cut on some edge of the tree and an {\bf admissible cut} on a tree consists of elementary cuts on some  edges of the tree such that any path starting from the root contains at most one of them. For such a cut $c$, the tree    $R^c(\mathfrak t )$
which contains the root of $\mathfrak t$ is called the trunk of the tree,
  and the product $P^c(\mathfrak t )$  of the remaining
trees, which is a forest,  is called the crown.
We define a partial order on trees by   $\mathfrak t^\prime\leq  \mathfrak t$ (we say $\mathfrak t^\prime$ is a subtree of $\mathfrak t$) if there is an admissible cut $c$ such
that $\mathfrak t^\prime=P^c(t )$.  It is compatible with the filtration since the subtree has fewer vertices than the original tree.
The filtered poset $\left(\mathcal{T},\leq \right) $ is connected since the empty tree is the only tree without vertices and it is clearly a subtree of any tree.
\end{ex}

\begin{ex}\label{ex:graph}   A {\bf Feynman graph} (see e.g. \cite{CK,M}) is a (non-oriented, non-planar) graph with a finite
number of vertices and edges. We shall assume that the edges (internal or external) are of  some given type which depends on the quantum field model we are considering, see e.g. \cite{M}.
A {\bf one-particle irreducible graph} (1PI graph) is a connected graph which remains connected when we cut an internal edge. The residue of a connected graph is the graph left over after shrinking all internal edges to a point.

The set $ \mathcal{F}  $ of  Feynman graphs   is filtered by the loop number
$L:= I-V+1$ where $I$ is the number of internal edges and $V$ is the number of vertices of a given graph.

For a connected graph $F$ in $\mathcal{F}$, we write $G\leq F$ if $G$ is  a subgraph of $F$, which should be $1$ PI if $F$ is $1$ PI.
This partial order is compatible with the filtration. However, the poset $\left( \mathcal{F},\leq\right)  $   is not connected since there are many graphs with zero loop number.
\end{ex}

\subsection{Complement maps on posets}
\begin{defn}\label{defn:complementposets} Let $(\mathcal{P}, \leq)$ be a poset such that for any $E\in {\mathcal P}$
\begin{equation}
 \label{eq:SE}
 \mathfrak{s}(E):=\{A\in \mathcal{P}\,|\, A\leq E\}  \end{equation} is   a finite set. A {\bf complement  } map on  $\calp$ assigns to any element  $ E\in {\mathcal P}$    a map
 \begin{eqnarray*}
\complement_E: \mathfrak{s}(E)& \longrightarrow & \mathcal{P}\\
 A&\longmapsto & E\backslash A
 \end{eqnarray*}
 satisfying the following properties
\begin{enumerate}
 \item {\bf Compatibility  with the partial order:} Let $A,C $ in $\mathcal{P}$ be such that $A\leq C$. Then
$$ \mathfrak{s}( C\backslash A)=\{B\backslash A\,|\, A\leq B\leq C\}.$$
\mlabel{it:ccom1}
 \item {\bf Transitivity:} Let $A, B, C$ in $ \mathcal{P}$ be such that $A\leq B\leq C$ \delete{and $|\mathfrak{s}(C)|<\infty$}. Then
 $$\left(C\backslash A\right)\backslash \left(B\backslash A\right)=C\backslash B.$$
\mlabel{it:ccom2}
\item {\bf  Compatibility with the filtration:}  Assume that the poset is filtered: ${\mathcal P}=\cup_{n\in \N } {\mathcal P}_n$. Then the complement map is   compatible with the  filtration in the sense that
    $$A\leq C\Longrightarrow \left\vert  C\backslash A\right\vert= \left\vert  C \right\vert-\left\vert   A\right\vert.$$
\mlabel{it:ccom3}
\item {\bf  Compatibility with the bottom:}  Assume that the poset is connected and let ${\mathcal P}_0=\{1\}$. Then
$$C\backslash 1= C\quad \text{for all } C\in {\mathcal P}.$$
\mlabel{it:ccom4}
\end{enumerate}
Condition (d) is obviously satisfied by previous examples of complement maps.
\end{defn}
\begin{remark}\label{rk:complementposets} If  the poset is connected, it follows from Condition (\ref{it:ccom3}) that  $C\backslash C=\{1\}$  for any $C\in {\mathcal P}$ since
 $$\left\vert C\backslash C\right\vert=\vert C\vert-\vert C\vert=0\Longrightarrow C\backslash C=\{1\}.$$
\end{remark}
Note that by (\mref{it:ccom1}), from $|\mathfrak{s}(C)|<\infty$ we have $|\mathfrak{s}(C\backslash A)|<\infty$ and $B\backslash A \in \mathfrak{s}\left(C\backslash A\right)$. Thus the expressions in (\mref{it:ccom2}) are well-defined.

\begin{ex}  Let $E$ be a   set.  For  $X\in \mathcal{P}_f(E),$  the complement set map:
\begin{eqnarray*}
\mathcal{P}_f(X)&\longrightarrow& \mathcal{P}_f(X)\\
Y        &\longmapsto    & X\setminus Y:=X\cap \overline{Y}
\end{eqnarray*} defines a complement map compatible with the filtration by the dimension.
\end{ex}

\begin{ex}\label{ex:complcone}As we saw in the previous section, the set complement map on $\ZZ _{\ge 0}$ induces a complement map on product cones which we recall here for convenience. Given a product cone $\langle e_I\rangle$ and a subset $J\subseteq I$, the map $\langle e_J\rangle\longmapsto  \langle e_{I\setminus J}\rangle$ defines a complement map on $  \mathcal{P}(\R^\infty)$ compatible with the filtration by the dimension of the cone.
\end{ex}

\begin{ex}  Let $E$ be a   separable Hilbert vector space equipped with a countable orthonormal basis $(e_1,\cdots, e_n,\cdots)$.  For any $I\subseteq \N$ we define   the set
$${\mathcal V}(I):=\{\lin( e_J)\,|\, J\in {\mathcal P}_f(I)\}$$
of finite dimensional vector subspaces  of $E$ spanned by  basis vectors indexed by finite subsets $J$ of $I$. The set ${\mathcal V}(\N)$ is equipped with a partial order given by the inclusion on the  index sets which is compatible with the filtration given by the cardinal of the index set. The map                                                                                                                                                                                                \begin{eqnarray*}                                                                                                                                    \mathcal{V}(I)&\longrightarrow& \mathcal{V}(I)\\                                                                                                                    \lin (e_J)         &\longmapsto    & \lin(e_I)\setminus \lin(e_J):=\lin(\pi_{\lin(e_J)^\perp}(e_I)\rangle= \lin( e_{I\setminus J})
                                                                                                                                                                  \end{eqnarray*}
defines a complement map compatible with the filtration by the dimension. Here $\pi_{\lin(e_J)^\perp}$ stands for the orthogonal projection onto the orthogonal complement to the linear space $\lin(e_J)$. The notation $\pi_{\lin(e_J)^\perp}(e_I)$ means that the projection  is applied to each basis vector indexed by an element of $I$.
                                                                                                                                                                                                \end{ex}
\begin{ex} This orthogonal complement map on linear spaces also induces the  complement map on product cones as can be seen from Eq.~(\ref{eq:perpcomplement}).
\end{ex}

\subsection{A complement map on  convex cones}
We now generalize the complement map built on product cones to general convex cones by means of   an orthogonal projection.

Let  $\mathcal{F}(C) $ be  the set of all faces of a  convex  cone $C\subseteq \R^k$.
We  borrow the following concept from \cite{BV} (see also \cite{GPZ3}) which we refer the reader to
for further details. The {\bf transverse cone}  to
$F\in \mathcal{F}(C)$ is   \begin{equation}\label{eq:transcone}t(C,F):=(C+\lin(F))/\lin(F),\end{equation} (where $\lin $ stands for the linear span) which we identify to the cone  in
 ${\mathcal C}\left(\R^\infty\right)$ defined by the projection $\pi_{F^\perp}(C)$ of $C$ onto the orthogonal complement\footnote{Our approach, like the one of Berline and Vergne in \cite{BV}, actually requires a choice a rational lattice  which consists of a pair     built from a cone   and a rational  lattice    in the linear space spanned by the cone. We refer the reader to \cite{GPZ3} for a detailed description. }
 $  {\rm lin} (F)^\perp $ in $\lin (C)$ for the canonical scalar product on $\R^\infty$.
\begin{ex} The transverse cone to a face $F=\langle e_J\rangle$ of a product cone $\langle e_I\rangle$ is the cone $\langle e_{I\setminus J}\rangle$, which corresponds to the  transverse cone $t\left( \langle e_I\rangle, \langle e_J\rangle\right). $
\end{ex}

\begin{ex}The  transverse cone to the face $F=\langle e_1+e_2\rangle$ in the cone $C=\langle e_1, e_1+e_2\rangle$ is the cone $t(C,F)=\langle  e_1-e_2 \rangle$.
Note that $t(C,F)$ is not a face of $C$.
 \end{ex}
 \begin{lemma}\label{lem:transversecone}    The map
 \begin{eqnarray*}
 \mathcal{F}(C)&\longrightarrow &\mathcal{C}(\R^\infty)\\
  F &\longmapsto& t(C,F)
  \end{eqnarray*}
  which to a face $F$ of a cone $C$ assigns the {\rm transverse cone} $t(C,F)$, is a complement map. More precisely, it enjoys the following properties.
  \begin{enumerate}
  \item {\bf  Compatibility with the partial order: }
  The set of faces of the cone $t(C,F)$ equals $\{t(G,F)\,|\, G \text{ a face of } C \text{ containing } F\}$.
   \label{it:com1}
  \item {\bf Transitivity:} $t(C,F)=t\left( t(C,F^\prime), t(F, F^\prime)\right)$ if $F^\prime$ is a face of $F$.
  \label{it:tra}
   \item {\bf  Compatibility with the dimension  filtration: }  ${\rm dim}(C)={\rm dim} (F)+{\rm dim}\left(t(C,F)\right)$ for any face $F$ of $C$.
       \label{it:com2}
       \item {\bf  Compatibility with the bottom:}  ${\mathcal C}$ is connected since  ${\mathcal C}_0$ is reduced to $1:=\{0\}$ and for any cone $C$ we have $t(C,\{0\})= C$.
\mlabel{it:com4}
\end{enumerate}
\end{lemma}

\begin{proof}
(\ref{it:com1}) Assume that $F$ is defined by the linear form  $u_F$ on $\R^\infty$, i.e.,
$$F= \{v\in C\ | \ \langle u_F,v \rangle=0\}.$$
Let $G$ be any face of $C$ containing $F$ that is defined by a linear form $u_G $ on $\R^\infty$, then $u_G|_F=0$.
Since a  linear form $u $ on $\R^\infty$ with $u|_F=0$ induces a linear form $u$ on $\linf$, we can view $u_G$ as a linear form on $ \linf$. It therefore defines a face $t(G,F)$ of $t(C,F)$. We can therefore define a map
$$t(\bullet, F):\{\text{faces  of } C \text{ containing }  F\}\to \{\text{faces of } t(C, F)\}
$$
$$G\mapsto t(G,F)=t(C,F)\cap \{v\in \R^\infty\ | \ \langle u_G,v \rangle=0\}.
$$

To check the bijectivity of $t(\bullet, F)$, we first note that any face of $t(C,F)$ is defined by some linear form  $u$ on $ \linf$ which can be viewed as a linear form on $\R^\infty$ that vanishes on ${\rm lin(F)}$. Hence $u$ defines a face $G$ of $C$ containing $F$. Thus $t(\bullet, F)$ is surjective.

For two different faces $G_1$, $G_2$ containing $F$  defined by linear forms $u_1$, $u_2 $ on $\R^\infty$, there are  vectors $v_1$ in $G_1$ and $v_2$ in $G_2$ such that $\langle u_1,v_2\rangle >0$ and $\langle u_2,v_1\rangle >0$. Thus $t(G_1,F)$ and $t(G_2,F)$ are different since the image of $v_1$ is not in $t(G_2, F)$ and the image of $v_2$ is not in $t(G_1,F)$. Hence the map $t(\bullet, F)$ is one-to-one.
\smallskip

\noindent
(\ref{it:tra}) The linear space $\lin(t(C,F))$ spanned by the transverse cone is the orthogonal space $\lin(F)^{\perp_{\lin(C)}}$ in $\lin(C)$ to   $ \lin(F)$. The transitivity then follows from the "transitivity" of the orthogonal complement map on linear spaces:
 $$\lin \left(t(t(C,F'),t(F,F'))\right)=\lin(t(F,F'))^{\perp_{\lin(t(C,F'))}}=\left(\left( \lin F^\prime\right)^{\perp_{\lin F}}\right)^{\perp_{\left(\lin(F^\prime)^\perp_{\lin C}\right)}}=\lin(F )^{\perp_{\lin C}}=\lin \left(t(F,C)\right).$$
 \smallskip

\noindent
(\ref{it:com2}) follows the fact that  $\lin (t(C,F))$ and $\lin (F)$ are orthogonal complements in $\lin (C)$.
 \end{proof}

\begin{ex} We use the notations of Example  \ref{ex:tree} above. See~\cite{F}  for further details.
 In $\mathcal{T} $, the   map $$\mathfrak t^\prime=P^c(\mathfrak t)\leq \mathfrak t \longmapsto  R^c(\mathfrak t)=\mathfrak t\setminus\mathfrak t^\prime $$ defines a complement map.  Let us first check the transitivity; let $\mathfrak t_3\preceq\mathfrak t_2\preceq\mathfrak t_1$, then cutting  the smaller trunk $\mathfrak t_3$ off both $\mathfrak t_2$ and $\mathfrak t_1$, before cutting off the remaining crown $\mathfrak t_2\setminus\mathfrak t_3$  off $\mathfrak t_1$ amounts to cutting off the whole  trunk $\mathfrak t_2$ from $\mathfrak t_1$. We now check the compatibility with the partial order; if $\mathfrak f$ is a forest made of trunks cut off from the crown $R^c(\mathfrak t)$ of a tree $\mathfrak t$ --i.e., if    $\mathfrak f\leq  \mathfrak f^\prime= R^c(\mathfrak t)=\mathfrak t\setminus P^c(\mathfrak t)$--  then there is a unique tree $\mathfrak t^\prime$ larger than $\mathfrak t$ i.e., $\mathfrak t\preceq\mathfrak t^\prime $, such that $\mathfrak f=\mathfrak t\setminus\mathfrak t^\prime$; $\mathfrak t^\prime$ is built from gluing $\mathfrak f^\prime$ as a crown onto  $P^c( \mathfrak t)$.
 \end{ex}

 \begin{ex} We use the notation of Example~\ref{ex:graph}. See~\cite{M} for further details.
 In $\mathcal{F} $, the complement $\Gamma\setminus \gamma$ of $\gamma\leq \Gamma$ in $\Gamma$ is the diagram obtained after "shrinking" the subdiagram $\gamma$ to a point. There is a bijection  $\gamma\mapsto \tilde \gamma= \gamma\setminus \delta$ from subgraphs of a graph $\Gamma$ containing $\delta$ onto subgraphs of $\Gamma\setminus \delta$ which shows  the compatibility of the complement map with the partial order.   The shrinking procedure is also clearly transitive $\Gamma\setminus \gamma= (\Gamma\setminus \delta)\setminus (\gamma\setminus \delta)$.
 \end{ex}

 \subsection{Coproducts derived from complement maps}
 Loosely speaking, coalgebras are   objects dual to algebras.  More precisely,  algebras  are dual to coalgebras but  the converse only holds in finite dimensions (see e.g. \cite{Ca}).

\begin{defn}
 A {\bf (counital) coalgebra} is a linear space ${\mathcal C}$ (here over $\R$) equipped with two linear maps:
\begin{enumerate}              \item a {\bf comultiplication} $\Delta:{\mathcal C}\to {\mathcal C}\otimes C$ written in  Sweedler's notation \cite{Sw}

$$  \Delta c=\sum_{(c)} c_{ (1)}\otimes c_{(2)}, $$ which  is
{\bf coassociative} $$(I\otimes \Delta)\otimes \Delta= ( \Delta\otimes I)\otimes \Delta.$$ The coassociativity of $\Delta$ translates to the following commutative diagram
$$
\xymatrix{ \calc \ar^\Delta[rr] \ar_\Delta[d] & & \calc\otimes \calc \ar^{I\otimes \Delta}[d]\\
\calc\otimes \calc \ar^{\Delta \otimes I}[rr] && \calc\otimes \calc\otimes \calc}
$$
and can be expressed in the following compact notation:
$$  \sum_{(c)}c_{(1)}\otimes\left(\sum_{(c_{(2)})}(c_{(21)}) \otimes (c_{(22)}) \right) = \sum_{(c)}\left( \sum_{(c_{(1)})}(c_{(11)}) \otimes (c_{(12)}) \right) \otimes c_{(2)}.$$
With Sweedler's notation \cite{Sw}, both these expressions read                                                                                                                                                                                                                                       $$ \sum_{(c)} c_{(1)}\otimes c_{(2)}\otimes c_{(3)}. $$                                                                                                                                                \item a {\bf counit } $\e: {\mathcal C}\to \R$  satisfying the {\bf counitarity} property

\begin{equation}
\label{eq:conunitarity}  (I_{\mathcal C}\otimes \e)\circ\Delta= ( \e\otimes I_{\mathcal C})\circ \Delta= I_{\mathcal C},\end{equation} with the identification ${\mathcal C}\otimes \R\simeq {\mathcal C}\simeq \R \otimes {\mathcal C}$. This translates to the following commutative diagram:
$$
\xymatrix{\calc\ot \calc \ar_{\e\ot I_\calc}[d] & \calc \ar_\Delta[l] \ar^\Delta[r] \ar_\e[d] & \calc\ot \calc \ar^{I_C\ot \e}[d]\\
\R\ot \calc\ar^{\cong}[r] & \calc & \calc\ot \R \ar_{\cong}[l] }
$$

The fact that $\e$ is a counit can   be expressed by means of the following formula

$$ c=\sum_{(c)} \varepsilon(c_{(1)})c_{(2)} = \sum_{(c)} c_{(1)}\varepsilon(c_{(2)}).$$
\end{enumerate}
The coalgebra is {\bf cocommutative} if $\tau\circ \Delta=\Delta$ where $\tau: {\mathcal C}\otimes {\mathcal C}\to {\mathcal C}\otimes {\mathcal C}$ is the flip  $c_1\otimes c_2\longmapsto c_2\otimes c_1$. This translates to the following commutative diagram:
$$
\xymatrix{
& \calc \ar_\Delta[dl] \ar^\Delta[dr] & \\
\calc\ot \calc \ar^\tau[rr] && \calc\ot \calc }
$$
and the equation
$$ \sum_{(c)} c_{(1)}\otimes c_{ (2)}= \sum_{(c)} c_{(2)}\otimes c_{(1)}.$$
The coalgebra ${\mathcal C}$ is {\bf  coaugmented} if there is  a morphism of coalgebras $ u: \R\to {\mathcal C}$ in which case we have $\e\circ u= I_\R$ and we set $1_C:=u(1_\R)$ where $1_\R$ is the unit  in $\R$. If
${\mathcal C}$ is coaugmented, then $ {\mathcal C}$ is canonically isomorphic to Ker $\e\oplus \R 1_{\mathcal C}$. The kernel Ker $\e$ is often denoted by
  $\overline  {\mathcal C}$ so
  $ {\mathcal C}= \overline {\mathcal C}\oplus \R 1_{\mathcal C}$.
   Let
   ${\mathcal C}
   =\R
   1_{\mathcal C }\oplus \overline
   {\mathcal C}$
    be a coaugmented coalgebra. The
    coradical filtration
     on
$ {\mathcal C}$
      is defined as follows:
Define $ {\mathcal F}_0 {\mathcal C} :=\R 1_{\mathcal C}$, and for $r\in \N$, we set
$$ {\mathcal F}_r{\mathcal C}
:= \R 1 \oplus \{x\in \overline{\mathcal C}\,|\, \overline \Delta^n x = 0\quad\forall n>r\}.$$
Here we have set $\overline \Delta x= \Delta x-  \left(1_C\otimes x+x\otimes 1_C\right)$ and $\overline \Delta^n $ is the   $n$-th iteration. A coalgebra ${\mathcal C}$ is said to be {\bf conilpotent} (or sometimes {\bf connected} in the literature) if it is coaugmented and if the filtration is exhaustive, that is
$  {\mathcal C}= \cup_{r\in \N}{\mathcal F}_r {\mathcal C}$.
\end{defn}
We are ready to build a coproduct from a complement map.

\begin{prop}\label{prop:complcoproduct}
Let a poset $\left( \mathcal{P},\leq\right)$  be  such that    for any $ E$ in $\mathcal {P}$ the set $\mathfrak{s}(E)$ defined as in Eq.~$($\ref{eq:SE}$)$ is finite and let it be equipped with a complement map, which assigns to any element  $ E\in {\mathcal P}$    a map
$$
\complement_E: \mathfrak{s}(E) \longrightarrow  \mathcal{P}
$$
$$
 A\longmapsto  E\backslash A.
$$
Then the map
\begin{eqnarray*}
 \Delta:\mathcal{P}&\longrightarrow& \mathcal{P}\otimes \mathcal{P}\\
 E   &\longmapsto    & \sum_{A\in \mathfrak{s}(E)} E\backslash A \otimes A,
 \end{eqnarray*}
extends linearly to a coassociative coproduct on the space $\K\,\mathcal{P}$  freely generated over a field $\K$ by $\mathcal{P}$.

If the poset is filtered $\mathcal{P}= \cup_{n\in \N}
{\mathcal P}_n $ and the complement  map is compatible with the filtration then so is the coproduct, that is, if
$C$ is in ${\mathcal P}_n$, then $\Delta C$ is in $\sum\limits_{p+q=n}   {\mathcal P}_p \otimes  {\mathcal P}_q.$

Let   $\e:{\mathcal P}\to \K$  be  zero outside ${\mathcal P}_0$ where it takes the value one and let us denote its linear extension to $\K\, {\mathcal P}$ by the same symbol. If moreover the poset ${\mathcal P} $ is connected, then
   the   linear space  $\left( \K\,\mathcal{P},\Delta,\e\right)$    is a counital  connected coalgebra.
\end{prop}

\begin{proof}
We first check the coassociativity.

Let $C\in \mathcal{P}$. On the one hand we have

$$
(I\otimes \Delta)\Delta(C) =  \sum_{B\leq C} (I \otimes \Delta) (C\setminus B \otimes B )
= \sum_{D\leq B\leq C}  C\setminus B \otimes  B\setminus D \otimes D.
$$

On the other hand,
$$\begin{array}{rcl}
(\Delta \otimes I)\Delta(C)&=&\sum_{D\leq C} ( \Delta \otimes I) (C\setminus D \otimes D)\\
& &\\
&=& \sum_{D\leq C}\sum_{H\leq C\setminus D}\left ((C\setminus D)\setminus H \right )\otimes H \otimes D\\
& &\\
&=&\sum_{D\leq B\leq C} (C\setminus D)\setminus ( B \setminus D)\otimes  B\setminus D \otimes D \quad {\text{(compatibility with the partial order)}}\\
& &\\
&=& \sum_{D\leq B\leq C}  C\setminus B \otimes  B\setminus D \otimes D.\quad {\text{(transitivity)}}
\end{array}$$

Let us check the counitarity.
For any $C\in {\mathcal P}$, setting ${\mathcal P}_0=\{1\}$ and using the fact that $C\backslash 1=C$ (see item (\ref{it:ccom4}) in Definition \ref{defn:complementposets}), for any $C\in {\mathcal P}$  we have
$$\sum_{B\leq C} \e(B) \, \left(C\backslash B\right)=\sum_{\vert B\vert=0}\e(B)\, \left(C\backslash B\right)=\e(1)\, \left(C\backslash 1\right)= 1_\K\cdot C=C.$$
Furthermore, since $\vert C\backslash B\vert=0\Longrightarrow  C\backslash B =1\Longrightarrow B=C$, using the fact that  $C\backslash C=1$ (see Remark  \ref{rk:complementposets}) we have
$$\sum_{B\leq C} B\,\e\left( C\backslash B\right) =\sum_{\left\vert   C\backslash B\right\vert=0}  B \, \e\left( C\backslash B\right)=C\,\e(C\backslash C)   =C\cdot  1_\K =C.$$
\end{proof}

\begin{ex} The  vector space $ \R{\mathcal P}_f(E) $ spanned by finite subsets of a  finite set $E$  defines a conilpotent   coalgebra.
\end{ex}

\begin{ex}\label{ex:coproductoncones} The free algebra  $\R {\mathcal C}(\R^\infty)$ spanned by closed convex cones  pointed at zero  in  $ \R^\infty$ defines a conilpotent   coalgebra.
\end{ex}

\begin{ex} The free algebra $\R\mathcal{T}  $ generated by planar rooted trees  defines a conilpotent   coalgebra.
\end{ex}

\section{Algebraic Birkhoff factorization on a conilpotent coalgebra}
\mlabel{sec:abd}
We give a generalization (\cite {GPZ3}) of  the algebraic Birkhoff factorization used for renormalization purposes in quantum field theory (see \cite{CK,M}) in so far as we weaken the assumptions on the source space which is not anymore assumed to be a Hopf algebra  but only a  coalgebra,  as well as   on the target algebra which is not anymore required to decompose into two subalgebras. We first define the convolution product and give its main properties.

\subsection{The convolution product}
Let $({\mathcal A},m_{\mathcal A},1_{\mathcal A})$ be an (unital) commutative algebra over  $\R$.
\begin{prop}\mlabel{prop:phiinv} $($see e.g. \cite[Proposition~II.3.1]{M}$)$ Let $({\mathcal C},\Delta_{\mathcal C},\e_{\mathcal C})$ be a  (counital) coalgebra over   $\R$.
\begin{enumerate}
\item The {\bf convolution} product on ${\mathcal L}({\mathcal C},{\mathcal A})$ defined as              $$\phi\ast \psi= m_{\mathcal A}\circ (\phi\otimes \psi)\circ \Delta_{\mathcal C}$$ is associative. In Sweedler's notation it reads:                             $$\phi\ast\psi(x)= \sum_{(x)} \phi(x_{(1)})\psi(x_{(2)}).$$
\mlabel{it:conv1}
\item  $e:=u_{\mathcal A}\circ\e_{\mathcal C}$  is a unit for the convolution product on                         ${\mathcal L}({\mathcal C},{\mathcal A})$.
\mlabel{it:conv2}
\end{enumerate}
\end{prop}

\begin{proof}
(\mref{it:conv1})
Using the coassociativity of $\Delta_{\mathcal C}$ and the associativity of $m_{\mathcal A}$ (we omit the explicit mention of the product in the computation below), for three  $\phi, \psi, \chi \in {\mathcal L}({\mathcal C},{\mathcal A})$ and using Sweedler's notations $\Delta x=\sum_{(x)}x_{(1)}\otimes x_{(2)}; \quad \Delta x_{(i)}=\sum_{(x_{(i)})}x_{(i1)}\otimes x_{(i2)}$ for each $i\in \{1,2\}$ we have
\begin{eqnarray*}\left(\left(\phi\ast\psi\right)\ast \chi\right)(x) &=&  \sum_{(x)}  \left( \phi(x_{(11)}) \psi  (x_{(12)})\right)\chi(x_{(2)})\\
&=&  \sum_{(x)}   \phi(x_{(11)}) \, \left(\psi  (x_{(12)}) \chi(x_{(2)})\right) \quad (\text{since}\quad m_{\mathcal A}\quad \text{is associative})\\
&=&  \sum_{(x)}  \phi(x_{(1) })\, \left(\psi  (x_{(21)})\chi(x_{(22)})\right)\quad (\text{since}\quad \Delta_{\mathcal C}\quad \text{is coassociative})\\
&=&  \sum_{(x)}  \phi(x_{(1)})\, \psi  (x_{(21)})\chi(x_{(22)}) \quad (\text{since}\quad m_{\mathcal A}\quad \text{is   associative})\\
&=&\left(\phi\ast \left(\psi\ast \chi\right)\right)(x).
\end{eqnarray*}
\smallskip

\noindent
(\mref{it:conv2})
Let $\phi\in {\mathcal L}\left({\mathcal C},{\mathcal A}\right)$.  Since $(\e_C\otimes I)\circ \Delta_C=I=(I\otimes \e_C)\circ \Delta_C $ we have    $$e\ast \phi(x)= \sum_{(x)}u_{\mathcal A}\circ\e_C( x_{(1)})\,  \phi( x_{(2)})= u_{\mathcal A} (1_C)\, \phi(x)=   \phi(x),$$
and similarly, we show that $\phi\ast e(x)= \phi(x)$ for any $x\in {\mathcal C}$.
\end{proof}

\begin{ex}\label{ex:convprodcones}
The convolution product of two maps $\phi$ and $\psi$  in  ${\mathcal L}\left({\mathcal P} (\R^\infty), {\mathcal A}\right)$  on a product cone $C=\langle e_I\rangle $ derived from the complement map described in Example~\ref{ex:complcone} reads
$$\phi\ast \psi \left(\langle e_I\rangle\right)= \sum_{J\subset I} \phi \left(\langle e_{I\setminus  J} \rangle\right) \psi \left(\langle e_J \rangle\right)=\sum_{F\in {\mathcal F}(C)}\phi(\overline F^C)\, \psi(F) $$ with the notation of  Eq.~({eq:perpcomplement}).

Setting ${\mathcal A}=\text{Mer}_{\rm sep} (\C^\infty)$, then  Eqs.~(\ref{eq:prodBHF}) and (\ref{eq:prodEML}) seen as   identities of    maps on product cones   read
\begin{equation}
\label{eq:convS} {\mathcal  S}=  {\mathcal  S}_+\ast  {\mathcal  S}_-= \mu \ast  {\mathcal  I}.\end{equation}
 \end{ex}

\begin{prop} \label{prop:connectedcoalgebra}$($\cite[Proposition II.3.1.]{M}$)$
Let ${\mathcal C}$ be a  connected augmented 
coalgebra and ${\mathcal A}$ an algebra. The set
$${\mathcal G}({\mathcal C},{\mathcal A}):=\{\phi\in {\mathcal L}({\mathcal C},{\mathcal A}),\quad\phi(1_{\mathcal C})=1_{\mathcal A}\}$$ endowed with the convolution product is a group with unit $e:=\e_{\mathcal C}\circ u_{\mathcal A}$ and   inverse
\begin{equation}
\label {eq:phistarinverse}\phi^{\ast (-1)}(x):=  \sum_{k=0}^\infty (e-\phi)^{\ast k}(x)
\end{equation} is well defined as a finite sum.
\end{prop}                                                                                                                     \begin{proof} We saw that    $e $ is a unit for the convolution product. Let us now build an inverse $\phi^{\ast (-1)}$ for any $\phi$ in $ \mathcal {G}\left({\mathcal C}, {\mathcal A}\right)$:                                                                                                            $$\phi^{\ast (-1)}(x)= \left(e-(e-\phi)\right)^{\ast (-1)}(x)=\sum_{k=0}^\infty (e-\phi)^{\ast k}(x).$$ Note that we indeed have                           $$\phi\ast\sum_{k=0}^\infty (e-\phi)^{\ast k}= (\phi-e)\ast\sum_{k=0}^\infty (e-\phi)^{\ast k}+ \sum_{k=0}^\infty (e-\phi)^{\ast k}=-\sum_{k=1}^\infty (e-\phi)^{\ast k}+ \sum_{k=0}^\infty (e-\phi)^{\ast k}= 1.$$                                                    On the one hand, since $\e(1_{\mathcal C})= u_{\mathcal A}\circ \e(1_{\mathcal C})= 1u_{\mathcal A}(1)= 1_{\mathcal A}$ we have $(e-\phi)(1_{\mathcal C})=0$ by assumption on $\phi$.  Hence  $(e-\phi)^{\ast k}(1_{\mathcal C})=0$ for $k>1$ so that the above formula yields
 $\phi^{\ast (-1)}(1_{\mathcal C})=1_{\mathcal A}$ and $\phi^{\ast (-1)}\in {\mathcal G}\left({\mathcal C}, {\mathcal A}\right)$.                                                                           On the other hand,  for any $x$ in ${\rm Ker}(\e)$ we have                                                                                                 $$ (e-\phi)^{\ast k}(x)= m_{{\mathcal A},k-1}(\phi\otimes\cdots\otimes\phi) \circ \overline \Delta^{k-1}(x).$$                                      Since element $x$ lies in some $\overline {\mathcal C}_n$,    this expression vanishes for $k\geq n+1$. Hence the above power series is finite and therefore defines an inverse of $\phi$ which lies in ${\mathcal G}$ since  $\phi^{\ast (-1)}(1_{\mathcal C})= (e-\phi)^{\ast 0}(x)=1_{\mathcal A}$
\end{proof}

\begin{ex} Back to Example \ref{ex:convprodcones}, we can rewrite the renormalized holomorphic part (\ref{eq:convS}) of $S$  as$$  {\mathcal  S}_+= {\mathcal S} \ast {\mathcal  S}_-^{\ast (-1)}= {\mathcal S} \ast  {\mathcal  I}^{\ast (-1)}.$$
\end{ex}

\subsection{Algebraic Birkhoff factorization}
We quote the following result from Theorem 3.2 in {GPZ3}.
\begin{thm}
Let ${\mathcal C} =\bigoplus_{n\geq 0} {\mathcal C}^{(n)}$ be a  connected coalgebra  and  let ${\mathcal A}$ be a unitary  algebra.  Let ${\mathcal A}={\mathcal A}_1\oplus {\mathcal A}_2$ be a linear decomposition  such that $1_{\mathcal A}\in {\mathcal A}_1$ and let $P$ be the induced projection onto ${\mathcal A}_1$ parallel to ${\mathcal A}_2$.\\
Given  $\phi\in {\mathcal G}({\mathcal C},{\mathcal A})$,
we define two maps $\varphi_i\in {\mathcal G}({\mathcal C},{\mathcal A}), i=1,2$ defined by the following recursive formulae on $\ker \e$:
\begin{eqnarray}
\varphi_1(x)&=&-P\left(\varphi(x)+\sum_{(x)} \varphi_1(x')\varphi(x'')\right),  \mlabel{eq:phi-}\\
\varphi_2(x)&=&(\id_A-P)\left(\varphi(x)+\sum_{(x)} \varphi_1(x')\varphi(x'')\right),
\mlabel{eq:phi+}
\end{eqnarray}
where, following Sweedler's notation,  we have set $\overline \Delta x=\sum x'\otimes x''$.
\begin{enumerate}
\item
We have  $\varphi_i(\ker \e)\subseteq {\mathcal A}_i$   and hence  $\varphi_i:{\mathcal C} \to \K 1_{\mathcal A} + {\mathcal A}_i$. Moreover, the following factorization holds
\begin{equation}
\varphi=\varphi_1^{\ast (-1)} \ast \varphi_2.
\mlabel{eq:abf}
\end{equation}
\mlabel{it:abf}
\item
$\varphi_i, i=1,2,$ are the unique maps in $ {\mathcal G}({\mathcal C},{\mathcal A})$  such that   $ \varphi_i(\ker \e)\subseteq {\mathcal A}_i$ for $i=1, 2,$  and satisfying Eqn. (\mref{eq:abf}).
\mlabel{it:uniq}
\item If moreover ${\mathcal A}_1$ is a subalgebra of ${\mathcal A}$ then $\phi_1 ^{\ast (-1)}$ lies in $ {\mathcal G}({\mathcal C},{\mathcal A}_1)$. \mlabel{it:phiinv}
\end{enumerate}
\mlabel{thm:abf}
\mlabel{thm:Birkhoff}
\end{thm}

\begin{proof}
(\mref{it:abf}) The proof is the same as  in \cite[Theorem 3.2]{GPZ3} ignoring the differential structure discussed there. We reproduce the proof for the sake of completeness.

The inclusion $\varphi_i(\ker \e)\subseteq {\mathcal A}_i, i=1,2,$  follows from the definitions. Further
\begin{eqnarray*}
\varphi_2(x)&=& (\id_{\mathcal A} -P)\left(\varphi(x)+\sum_{(x)}\varphi_1(x')\varphi(x'')\right)\\
&=& \varphi(x) +\varphi_1(x) +\sum_{(x)}\varphi_1(x')\varphi(x'')\\
&=& (\varphi_1\ast \varphi)(x).
\end{eqnarray*}
Since $\varphi_1(\I)= 1_{\mathcal A}$, $\varphi_1$ is invertible for the  convolution product in ${\mathcal A}$ as a result of Proposition  \ref{prop:connectedcoalgebra}  applied to   $\varphi_1$, from which  Eq.~(\mref{eq:abf}) then follows.

\noindent
(\mref{it:uniq}) Suppose there are $\psi_i\in {\mathcal G}\left({\mathcal C},{\mathcal A}\right), i=1,2,$ with $  \psi_i(\ker \e)\subseteq {\mathcal A}_i$ such that $\varphi=\psi_1^{\ast (-1)} \ast \psi_2 $ and $\psi_i (J)=1_{\mathcal A}$.
We prove $\varphi_i(x)=\psi_i(x)$ for $i=1,2, x\in {\mathcal C}^{(k)}$ by induction on $k\geq 0$. These equations hold for $k=0$. Assume that the equations hold for ${\mathcal C}^{(k)}$. For $x\in {\mathcal C}^{(k+1)}\subseteq \ker(\e)$, by $  \varphi_2=\varphi_1\ast \varphi$ and  $  \psi_2= \psi_1 \ast \varphi,$  we have
$$\varphi_{2}(x)= \varphi_1(x)+ \varphi(x)+ \sum_{(x)}\varphi_1(x^\prime)\varphi(x^{\prime \prime}) \text{ and
}
\psi_{2}(x)= \psi_1(x)+ \varphi(x)+ \sum_{(x)}\psi_1(x^\prime)\varphi(x^{\prime \prime}),$$ where we have used
$\varphi_1(\I)=\psi_1(\I)=\varphi(\I)=1_A$ . Hence by the induction hypothesis, we have
$$
\varphi_{2}(x)-\psi_{2}(x)=\varphi_{1}(x)-\psi_{1}(x)+\sum_{(x)}(\varphi_{1}(x^{\prime  })-\psi_{1} (x^{\prime  })) \varphi(x^{\prime \prime})=\varphi_{1}(x)-\psi_{1} (x)\in A_{1}\cap A_2=\{0\}.$$
Thus
$$\varphi_i(x)=\psi_i (x) \text{ for all } x\in \ker(\e), i=1,2.$$

\noindent (\mref{it:phiinv}) If   ${\mathcal A}_1$ is a subalgebra, then it follows from Proposition \mref{prop:phiinv} applied to ${\mathcal A}_1$ instead of ${\mathcal A}$, that
$\varphi_1$ is invertible in ${\mathcal A}_1$.
\end{proof}

\section{Application to renormalized conical zeta values}

\subsection{Algebraic Birkhoff factorization on cones}

The algebraic Birkhoff factorization   can be carried out from exponential sums on product cones to exponential sums on  general convex polyhedral cones using the complement map described in Lemma \ref{lem:transversecone}   built from the transverse cone to a  face \cite{GPZ3}, which generalizes the orthogonal complement used in the case of product cones.
Here we consider both   closed convex (polyhedral) cones in $\R^k$
$$
 \cone{v_1,\cdots,v_n}:=\RR\{v_1,\cdots,v_n\}=\RR _{\geq 0}v_1+\cdots+\RR_{\geq 0}v_n,
$$
 where $v_i\in \R^k$, $i=1,\cdots, n$ defined previously and
  open cones defined in a similar manner replacing $\R_{\geq 0}$ by $\R_+$.
Product cones $\langle e_i, i\in I\rangle$ with   $I\subset \{1,\cdots, k\}$  and $\{e_i\,|\,i\in \{1,\cdots,k\}\}$ the canonical basis of $\R^k$ are convex cones. We shall focus here on  Chen cones $\langle e_{i_1}, e_{i_1}+e_{i_2},\cdots, e_{i_1}+\cdots+ e_{i_n}  \rangle$ with $\{i_1,\cdots, i_n\} \subset \{1,.\dots, k\}$, which  are   closed convex cones as well as their open counterparts.

Both the complement map defined by means of the transverse map in Lemma~\ref{lem:transversecone}  and  the corresponding coproduct  defined in Proposition \ref{prop:complcoproduct} (see also Example \ref{ex:coproductoncones}) are  compatible with subdivisions in a suitable sense. Recall that
 a {\bf subdivision} of a cone $C$ is a set $\{C_1,\cdots,C_r\}$ of cones such that
 \begin{enumerate}
 \item[(i)] $C=\cup_{i=1}^r C_i$,
 \item[(ii)] $C_1,\cdots,C_r$ have the same dimension as $C$ and
 \item[(iii)] $C_1,\cdots,C_r$ intersect along their faces, i.e.,  $C_i\cap C_j$ is a face of both $C_i$ and $C_j$.
 \end{enumerate}
 \begin{ex} The product cone $\langle e_1,e_2\rangle$ can be subdivided into two Chen cones $\langle e_1,e_1+e_2\rangle$ and $\langle e_1+e_2,e_2\rangle$.
 \end{ex}
  To a simplicial convex (closed) cone $C\subset \Z^k$, namely one whose generators are linearly independent, one can assign an exponential sum and an exponential integral  which can informally be described as follows
 $$
 S^c (C)(\vec \e ): =
 \sum_{\vec{n}\in C\cap \\Z^k} e^{\langle \vec n, \vec \e \rangle}
 ;\quad S^o(C)(\vec \e ): =
 \sum_{\vec{n}\in C^o\cap \Z^k} e^{\langle \vec n, \vec \e \rangle};
\quad I(C)(\vec \e)=\int_{  C } e^{\langle \vec x, \vec \e\rangle} \, d\vec x.$$
Here  $C^o$ is  the open cone given by  the interior of $C$ and  $\e$ is taken in
$$\check C^-_k:=\left\{\left.\vec \e:=\sum_{i=1}^k \e_ie_i^\ast \,\right|\, \langle \vec x, \vec \e \rangle <0 \text{ for all } \vec x\in C\right\},$$
where $\{e_i^\ast\,|\, i\in \{1,\dots, k\}\}$ is the dual canonical basis and $\langle \vec x, \vec \e\rangle$ the natural pairing $R^k\otimes  \left(\R^k\right)^\ast \to \RR$.
Keep in mind that a precise formulation requires introducing a lattice attached to the cone, so considering lattice cones instead of mere cones  (see \cite{GPZ3}).  This then extends to any convex cones by additivity on subdivisions.

Whereas exponential sums on product cones take their values on products of meromorphic functions in one variable, exponential sums on general convex cones take  their values in the larger space of meromorphic maps with simple linear poles supported by the faces of the cone.

It turns out that  a meromorphic map function with linear poles also decomposes as a sum of a holomorphic  part and a polar part. A decomposition of the algebra of   germs of meromorphic functions with linear poles  into the holomorphic part and a linear complement   was shown in \cite{GPZ4}    by means of  an inner product using our results on cones and pure fractions in an essential way. We shall denote by $\pi_+$ the corresponding projection onto the holomorphic part.

Consequently, one can implement an algebraic Birkhoff factorization    \cite{GPZ3} on the coalgebra of convex polyhedral cones. \footnote{We actually carry out the algebraic Birkhoff factorization   on lattice cones.} Just as the algebraic Birkhoff factorization   gave rise to an Euler-Maclaurin formula on product cones, when the inner product used to defined the coproduct on cones coincides with the inner product used to decompose the space of meromorphic germs,  the  algebraic Birkhoff factorization   of the exponential sum on a convex (lattice)  cone yields back Berline and Vergne's local Euler-Maclaurin formula   \cite{BV}.
To prove this identification which is easy to see on smooth cones, we subdivide  a general convex cone   into simplicial ones and  use  the compatibility of $S_-$ in the factorization procedure with subdivisions. This compatibility is shown by means of a rather involved combinatorial proof.

Recall  from  Eq.~(\ref{eq:productmutltizeta}) that the ``holomorphic part" of the exponential discrete sum on product cones generates products of renormalized zeta values at non-positive integers  as coefficients of its Taylor expansion at zero. Similarly \cite{GPZ3}, the "holomorphic part" of the exponential discrete sums on general convex polyhedral cones obtained from an algebraic Birkhoff factorization, generates what we call renormalized {\em conical} zeta values at non-positive integers which arise as coefficients of its Taylor expansion at zero. It turns out that the "holomorphic part"  of the exponential sums  $S^c(C)$ and $S^o(C)$  on a cone $C$ derived from the algebraic Birkhoff factorization   actually coincides with  the projection $\pi_+(S^c(C))$ and $\pi_+(S^o(C))$, when the inner product used to defined the coproduct on cones coincides with the inner product used to decompose the space of meromorphic germs respectively, onto their holomorphic part when seen as   meromorphic functions with linear poles.

\subsection{Meromorphic functions with linear poles}
For later use, we define  the projection $\pi_+$ somewhat informally; a precise definition can be found in \cite{GPZ4}. One shows  that  a meromorphic function $f=\frac{h}{L_1\cdots L_n}$ on $\C^k$ with linear poles $L_i, i=1,\cdots, n$ given by  linear forms   and $h$ a holomorphic function at zero, uniquely decomposes as
\begin{equation} f=\sum_{i=1}^n\left(\frac{h_i(\vec{\ell}_i)}{\vec L_i^{\vec s_i}}+\phi_i(\vec \ell_i, \vec L_i)\right),
\mlabel{eq:sumi}
\end{equation}
with $|\vec s_i|> 0$ and where $\vec L_i=(L_{i1}, \cdots , L_{im_i})$, $\{L_{i1}, \cdots , L_{im_i}\}$ is a linear independent subset of $\{L_1,\cdots,L_n\}$, extended to a basis $\{ \vec L_i,\vec \ell_i\}$ of $\C^k$, with $\vec \ell_i=(\ell _{i(m_i+1)}, \cdots \ell _{ik})$, $L_{ij}$, $\ell _{im}$ orthogonal for the canonical inner product  on $\C^k$ and $h_i(\vec{\ell}_i)$ holomorphic (reduced to a constant when $k=1$). Then  we call $f_+:=\pi_+(f)=\sum\limits_{i=1}^n\phi_i$, which is a germ of holomorphic function in the independent variables $\vec \ell_i$ and $\vec L_i$,  the {\bf holomorphic part} of $f$ and $f_-:=(1-\pi_+)(f)=\sum\limits_{i=1}^n \frac{h_i(\vec{\ell}_i)}{\vec L_i^{\vec s_i}}$  the {\bf polar part} of $f$.

In order to discuss examples, it is convenient to set the following
notation. Given $k$ linear forms $L_1, \cdots, L_k$, we set

\begin{equation}
 \label{eq:notationL}[L_1, \cdots , L_k]:=\frac {e^{L_1}}{1-e^{L_1}}\frac {e^{L_1+L_2}}{1-e^{L_1+L_2}}\cdots \frac {e^{L_1+L_2+\cdots +L_k}}{1-e^{L_1+L_2+\cdots +L_k}}.
\end{equation}
So, for any (closed) Chen cone $C_k=\langle e_1,e_1+e_2,\cdots,e_1+\cdots+e_k\rangle$ (here $e_1,\dots, e_k$ is the canonical basis of $\R^k$), we have $$S^o(C_k )(\e _1, \e _2,\cdots, \e_k) =[\e _1,\e _2,\cdots,\e _k].$$

\begin{ex}
\begin{enumerate}
\item Take $k=1$.  Let $f(\e)=\frac {e^{\e}}{1-e^{\e}}=  \frac {1}{ e^{-\e}-1}=-\frac{\rm Td(-\e)}{-\e}$ on $\C$. Then  by Eq.~(\ref{eq:SBernoulli}) we have
\begin{equation} \label{eq:dim1meroexp} f(\e)= -  \frac{1}{\e}- \frac{1}{2}-\sum_{k=1}^K \frac{B_{2k }}{(2k )!}\e^{2k-1} + o(\e^{2K}) =-\frac{1}{\e}+\phi(\e),
\end{equation}
with
\begin{equation}\label{eq:phi}\phi(\e):= - \frac{1}{2}-\sum_{k=1}^K \frac{B_{2k }}{(2k )!}\e^{2k-1} + o(\e^{2K})=-\frac 12-\frac 1{12}\e+\frac 1{720}\e^3+\cdots
\end{equation}
holomorphic at zero so $  \pi_+(f)=  \phi(\e)$.

\item Let $k=2$ and let $f(\e)= [\e_1, \e_1+\e_2]$. Applying Eq.~(\ref{eq:dim1meroexp}) we write
\begin{eqnarray}\label{eq:k2}\pi_+\left([\e_1, \e_2] \right)    &=&\pi_+\left(\Big(-\frac 1{\e_1 }+\phi(\e _1 )\Big)\Big(-\frac 1{\e_1+\e_2}+\phi(\e _1+\e _2)\Big)\right)\nonumber\\
&=& \pi_+\left(-\frac {\phi(\e _1+\e_2 )}{\e_1 }-\frac {\phi(\e _1 )}{\e_1+\e _2 }+\phi(\e _1 )\phi(\e _1+\e_2)\right)\nonumber\\
&=& -\frac {\phi(\e _1+\e_2 )-\phi(\e _2)}{\e_1}-\frac {\phi(\e _1 )-\phi\left(\frac {\e_1-\e_2}2\right)}{\e_1+\e _2 }+\phi(\e _1 )\phi(\e _1+\e_2 ),\notag
\end{eqnarray}
\item Let $k=3$ and let $f(\e)= [\e_1+\e_3, \e_2]$. Using Eq.~(\ref{eq:dim1meroexp}) we write
\begin{eqnarray*}\label{eq:k3ex1}
\pi_+\left([ \e_1+\e_3, \e_2 ]    \right) &=&\pi_+\left(\Big(-\frac 1{\e_1+\e_3}+\phi(\e _1+\e _3)\Big)\Big(-\frac 1{\e_1+\e_2+\e_3}+\phi(\e _1+\e_2+\e _3)\Big)\right)\nonumber\\
&=& \pi_+\left(-\frac {\phi(\e _1+\e_2+\e _3)}{\e_1+\e_3}-\frac {\phi(\e _1+\e _3)}{\e_1+\e _2+\e_3}+h(\e _1+\e _3) \,\phi(\e _1+\e_2+\e _3)\right)\nonumber\\                                        &=& - \frac {\phi(\e _1+\e_2+\e _3)-\phi(\e _2)}{\e_1+\e_3}-\frac {\phi(\e _1+\e _3)-\phi\left(\frac {\e_1+\e_3-2\e_2}3\right)}{\e_1+\e _2+\e_3}\\
&&+\phi(\e _1+\e _3)\phi(\e _1+\e_2+\e _3),
\end{eqnarray*}
since $\e_1+\e_3-2\e_2\perp \e_1+\e_2+ \e_3$. Similarly, for $f(\e)= [\e_1, \e_2+\e _3]$, we have
\begin{eqnarray}\label{eq:k3ex2} \pi_+\left([\e_1, \e_2+\e _3]\right)&=&\pi_+\left(\Big(-\frac 1{\e_1}+\phi(\e _1)\Big)\Big(-\frac 1{\e_1+\e_2+\e_3}+\phi(\e _1+\e_2+\e _3)\Big)\right)\nonumber\\
&=& \pi_+\left(-\frac {\phi(\e _1+\e_2+\e _3)}{\e_1}-\frac {\phi(\e _1)}{\e_1+\e _2+\e_3}+\phi(\e _1)\phi(\e _1+\e_2+\e _3)\right)\nonumber\\
&=&  -\frac {\phi(\e _1+\e_2+\e _3)-\phi(\e _2+\e_3)}{\e_1}-\frac {\phi(\e _1)-\phi\left(\frac {2\e_1-\e_2-\e_3}3\right)}{\e_1+\e _2+\e_3}+\phi(\e _1)\phi(\e _1+\e_2+\e _3).
\end{eqnarray}
\end{enumerate}
\end{ex}
\subsection{Renormalized conical zeta values: the case of Chen cones}
To   a (closed) convex polyhedral cone $C$,  one can assign   closed (resp. open)  renormalized conical zeta values $\zeta^c(C; -a_1,\cdots, -a_k)$ (resp. $\zeta^o(C;-a_1,\cdots, -a_k)$) with  $a_i\in \Z_{\geq 0}$, $i\in \{1,\cdots, k\}$  corresponding to     the   coefficient in $\e^{a_1}\cdots \e^{a_k}$ of the Taylor expansion at zero of $\pi_+\left(S^c(C)\right)$ (resp. $\pi_+\left(S^o(C)\right)$).

This applied to a (closed) Chen cone $C_k=\langle e_1,e_1+e_2,\cdots,e_1+\cdots+e_k\rangle$   gives rise to multiple  zeta values $$\zeta(-a_1,\cdots, -a_k):=\zeta^o(C_k;-a_1,\cdots, -a_k)$$    given by the Taylor coefficient in $\e^{a_1}\cdots \e^{a_k}$ of $\pi_+\left(S^o(C_k)\right)$  and multiple  zeta-star values
$$\zeta^\star(-a_1,\cdots, -a_k):=\zeta^c(C_k;-a_1,\cdots, -a_k)$$   given by the Taylor coefficient in $\e^{a_1}\cdots \e^{a_k}$ of $\pi_+\left(S^c(C_k)\right)$. The latter are algebraic expressions in the former.

With the notations of Eq.~(\ref{eq:notationL}) we have
\begin{equation}\label{eq:zetapi} \zeta  (-a_1, -a_2,\dots, -a_k)=\pi_+\left(\partial^{a_1}_1\partial^{a_2}_2\cdots \partial^{a_k}_k[\e _1,\e_2,\cdots, \e_k]\right)_{\vert_{\vec\e_=\vec 0}}.
 \end{equation}

\begin{ex}
 \begin{enumerate}
 \item \begin{equation}\label{eq:zeta1}
 \zeta\left(-a  \right) =\phi^{(a)}(0  )=-\frac{B_{a+1}}{a+1}.\end{equation}
 \item \begin{eqnarray}\label{eq:zeta12}
&&\zeta\left(-a_1,-a_2 \right)  \nonumber\\               &=& \left(\partial^{a_1}_1\partial^{a_2}_2\Big(-\frac {\phi(\e _1+\e_2 )-\phi(\e _2)}{\e_1}-\frac {\phi(\e _1 )-\phi\left(\frac {\e_1-\e_2}2\right)}{\e_1+\e _2 }+\phi(\e _1 )\phi(\e _1+\e_2 )\Big)\right)_{\vert_{\e_1=\e_2=0}}.                                                        \end{eqnarray}
 \item
\begin{eqnarray}\label{eq:zeta1plus23}
\zeta\left(-a_1-a_2, -a_3\right)&=& \left(\partial^{a_1}_1\partial^{a_2}_2\partial^{a_3}_3 \Big(\pi_+[\e_1+\e_3, \e_2]\Big)\right)_{\vert_{\vec \e=\vec 0}}\nonumber \\
 &=& \left(\partial^{a_1+a_3}_1\partial^{a_2}_2\Big(-\frac {\phi(\e _1+\e_2+\e _3)-\phi(\e _2)}{\e_1+\e_3}+\phi(\e _1+\e _3)\phi(\e _1+\e_2+\e _3)\Big)\right. \nonumber\\
  &+&\left.\partial^{a_1+a_3}_1\p ^{a_2}_2\Big(-\frac {\phi(\e _1+\e _3)-\phi\Big(\frac {\e_1+\e_3-2\e_2}3\Big)}{\e_1+\e _2+\e_3}\Big)\right)_{\vert_{\vec \e=\vec 0}}\nonumber\\
&=& \left( \partial^{a_1+a_3}_1\partial^{a_2}_2\Big(-\frac {\phi(\e _1+\e_2)-\phi(\e _2)}{\e_1}+\phi(\e _1)\phi(\e _1+\e_2)\Big)\right.\nonumber \\
 &+&\left.\partial^{a_1+a_3}_1\p ^{a_2}_2\Big(-\frac {\phi(\e _1 )-\phi\Big(\frac {\e_1 -2\e_2}3\Big)}{\e_1+\e _2 }\Big) \right)_{\vert_{\e_1=\e_2=0}}.
\end{eqnarray}
Similarly,
\begin{eqnarray}\label{eq:zeta12plus3} \zeta\Big(-a_1,-a_2  -a_3\Big)&=&\left(\partial^{a_1}_1\partial^{a_2}_2\partial^{a_3}_3[\e_1,\e_2+\e_3]\right)_{\vert_{\vec \e=\vec 0}}\nonumber \\
&=& \left( \partial^{a_1}_1\p ^{a_2+a_3}_2 \Big(-\frac {\phi(\e _1+\e_2+\e _3)-\phi(\e _2+\e_3)}{\e_1}+\phi(\e _1)\phi(\e _1+\e_2+\e _3)\Big)\right.\nonumber \\
&+&\left.\partial^{a_1}_1\partial^{a_2+a_3}_2\Big(-\frac {\phi(\e _1)-\phi\Big(\frac {2\e_1-\e_2-\e_3}3 \Big)}{\e_1+\e _2+\e_3}\Big)\right)_{\vert_{\vec \e=\vec 0}}\nonumber\\
&=& \left(\partial^{-s_1}_1\partial^{-s_2-s_3}_2\Big(-\frac {\phi(\e _1+\e_2)-\phi(\e _2)}{\e_1}+\phi(\e _1)\phi(\e _1+\e_2)\Big) \right.\nonumber\\
&+&\left. \partial^{a_1}_1\partial^{a_2+a_3}_2\Big(-\frac {\phi(\e _1)-\phi\Big(\frac {2\e_1-\e_2}3\Big)}{\e_1+\e _2}\Big)\right)_{\vert_{ \e_1=\e_2=\vec 0}}.
\end{eqnarray}
\end{enumerate}
\end{ex}

As we pointed out in the introduction, our geometric approach  contrasts with other approaches such as  \cite{GZ,MP} to the renormalization of multiple zeta values at non-integer arguments, where the algebraic Birkhoff factorization   is carried out on the summands (functions $(x_1,\cdots, x_k)\mapsto x_1^{-s_1}\cdots x_1^{-s_k}$) rather than on the domain (the cones) of summation as in our present construction or \cite{Sa} where a purely analytic renormalization method is implemented, which does not use algebraic Birkhoff factorization. The generality and the geometric nature of our approach nevertheless have a cost; whereas in \cite{GZ,MP} the  renormalized multiple zeta values obey the stuffle relations, these are not preserved  in our approach.
They  nevertheless do hold for two arguments; as one could expect from the relation
\begin{equation}\label{eq:Chen2}[\e _1][\e _2] =[\e_1,\e_2]+[\e _2, \e_1]+[\e _1+ \e_2 ],
 \end{equation}   for any non-positive integers $s_1,s_2$ we have
 $$\zeta(s_1)\,\zeta(s_2)= \zeta(s_1,s_2)+\zeta(s_2,s_1)+ \zeta(s_1+s_2).$$ Similarly,$$\zeta^*(s_1)\,\zeta^*(s_2)= \zeta^*(s_1,s_2)+\zeta^*(s_2,s_1)- \zeta^*(s_1+s_2).$$
However, as we shall see below, they do not necessarily hold for three arguments.

Table~\ref{tb:gpz1} displays values for the   double zeta values $\zeta(-a, -b)$ with $a,b\in \Z_{\neq 0}$ derived using our conical approach from the {\it open Chen cone}  $\langle e_1,e_1+e_2\rangle$ in $\R^2$, where $\{e_1,e_2\}$ is    the canonical basis of $\R^2$. It is then followed by Table~\mref{tb:gpz2} of values for the conical double zeta $\star$-values $\zeta^*(-a, -b)$ with $a,b\in \Z_{\neq 0}$ corresponding to the {\it closed Chen cone}  $\langle e_1,e_1+e_2\rangle$ in $\R^2$. The values in boldface of the first table coincide with the ones obtained in \cite{GZ,MP}.

 \begin{table}
 \centering
 \caption{Values of renormalized conical double zeta values}
 \vspace{-.7cm}
 {\small
  $$
  \begin {array} {|c|c|c|c|c|c|c|}
  \hline \zeta (-a_1,-a_2)&a_1 = 1& a_1 = 2 &a_1 = 3& a_1 = 4& a_1 = 5&a_1 = 6\\
  \hline a_2 = 1&  {\frac 1{288}}&   {-\frac 1{240}}& \frac {101}{80640}&   {\frac 1{504}}& -\frac
  {169}{96768}&   {-\frac 1{480}}
  \\
  \hline a_2 = 2&  {-\frac 1{240}}& 0&  {\frac 1{504}}&-\frac {7127}{9676800}&
   {-\frac 1{480}}&\frac {7097}{3870720}\\
  \hline a_2 = 3&-\frac {157}{80640}&   {\frac 1{504}}&
   {\frac 1{28800}}&   {-\frac 1{480}}& \frac {1543}{1892352}&
   {\frac 1{264}}
  \\
  \hline a_2 = 4&  {\frac 1{504}}& \frac {7127}{9676800}&   {-\frac 1{480}}&   {0}&
    {\frac 1{264}}& -\frac {9280679}{5960908800}
  \\
  \hline a_2 = 5& \frac {67}{32256}&
    {-\frac 1{480}}& -\frac {72251}{85155840}&
  {\frac 1{264}}&   {\frac 1{127008}}&  {-\frac {691}{65520}}\\
  \hline a_2 = 6&  {-\frac 1{480}}& -\frac {7097}{3870720}&  {\frac 1{264}}& \frac
  {9280679}{5960908800}&
    {-\frac {691}{65520}}& {0}\\
  \hline
  \end {array}
  $$
  }
  \mlabel{tb:gpz1}
  \vspace{.5cm}
  \centering
  \caption{Values of conical double zeta star values}
  \vspace{-.7cm}
  {\small
  $$
  \begin {array} {|c|c|c|c|c|c|c|}
  \hline \zeta ^\star(-a_1,-a_2)&a_1 = 1& a_1 = 2 &a_1 = 3& a_1 = 4& a_1 = 5&a_1 = 6\\
  \hline a_2 = 1&\frac 1{288}& \frac 1{240}& \frac {101}{80640}& -\frac 1{504}& -\frac {169}{96768}& \frac 1{480}
  \\
  \hline a_2 = 2&\frac 1{240}& 0& -\frac 1{504}& -\frac {7127}{9676800}& \frac 1{480}& \frac {7097}{3870720}
  \\
  \hline a_2 = 3&-\frac {157}{80640}& -\frac 1{504}& \frac 1{28800}& \frac 1{480}& \frac {1543}{1892352}& -\frac 1{264}
  \\
  \hline a_2 = 4&-\frac 1{504}& \frac {7127}{9676800}& \frac 1{480}& 0& -\frac 1{264}& -\frac {9280679}{5960908800}
  \\
  \hline a_2 = 5&\frac {67}{32256}& \frac 1{480}& -\frac {72251}{85155840}& -\frac 1{264}& \frac 1{127008}& \frac {691}{65520}
  \\
  \hline a_2 = 6&\frac 1{480}& -\frac {7097}{3870720}& -\frac 1{264}& \frac {9280679}{5960908800}& \frac {691}{65520}& 0
  \\
  \hline
  \end {array}
  $$
  }
  \mlabel{tb:gpz2}
  \end{table}

 Stuffle relations hold for two arguments but fail to hold for three arguments.
 For example, from  the following relation
 \begin{equation}\label{eq:Chen3}[\e _1,\e _2][\e _3]=[\e _3,\e_1,\e_2]+[\e _1, \e _3, \e_2]+[\e _1, \e_2, \e_3]+[\e_1+\e_3, \e_2]+[\e_1,\e_2+\e_3],
 \end{equation}
  one might expect the following stuffle relation
\begin{eqnarray}
 \label{eq:stuffle3}
 \zeta (-a_1, -a_2)\zeta (-a_3)&=& \zeta (-a_3, -a_1, -a_2)+\zeta (-a_1, -a_3, -a_2)+\zeta (-a_1, -a_2, -a_3)\\
 &&+\zeta (-a_1-a_3, -a_2)+\zeta (-a_1, -a_2-a_3) \notag
 \end{eqnarray}
 to hold.
 \begin{prop}
The stuffle relation $($\ref{eq:stuffle3}$)$   is  violated  for some values  $(a_1,a_2,0)\in \Z_{\geq 0}^3$ with $a_1+a_2>2$.
 \end{prop}
 \begin{proof}
 Were the stuffle relation (\ref{eq:stuffle3}) to hold for any $a_1, a_2, a_3\in \ZZ _{\ge 0}$,  combining Eqs.~(\ref{eq:zeta1}),  (\ref{eq:zeta12}), (\ref{eq:zeta1plus23}) and (\ref{eq:zeta12plus3}), we would have
 \begin{eqnarray*}&& \left(\pi_+\Big(\partial^{a_1+a_3}_1\partial^{a_2}_2\Big(\frac {\phi(\e _1)-\phi\Big(\frac {\e_1-2\e_2}3\Big)}{\e_1+\e _2}\Big)+\partial^{a_1}_1\partial^{a_2+a_3}_2\Big(\frac {\phi(\e _1)-\phi\Big(\frac {2\e_1-\e_2}3\Big)}{\e_1+\e _2}\Big)\Big)\right)_{\vert_{\vec \e=\vec 0}}\\
&=& \left(\pi_+\Big(\partial^{a_1+a_3}_1\partial^{a_2}_2\Big(\frac {\phi(\e _1)-\phi\Big(\frac {\e_1-\e_2}2\Big)}{\e_1+\e _2}\Big)+\partial^{-s_1}_1\partial ^{-s_2-s_3}_2\Big(\frac {\phi(\e _1)-\phi\Big(\frac {\e_1-\e_2}2\Big)}{\e_1+\e _2}\Big)\Big)\right)_{\vert_{\vec \e=\vec 0}}
\end{eqnarray*}
for any $a_1, a_2, a_3\in \ZZ _{\ge 0}$, in particular, for $a_3=0$, in which  case
\begin{eqnarray*}&& \left(\pi_+\Big(\partial^{a_1}_1\partial^{a_2}_2\Big(\frac {\phi(\e _1)-\phi\Big(\frac {\e_1-2\e_2}3\Big)}{\e_1+\e _2}+\frac {\phi(\e _1)-\phi\Big(\frac {2\e_1-\e_2}3\Big)}{\e_1+\e _2}\Big)\Big)\right)_{\vert_{\vec \e=\vec 0}}\\&=&  \left(\pi_+\Big(\partial^{a_1}_1\partial^{a_2}_2\Big(\frac {\phi(\e _1)-\phi\Big(\frac {\e_1-\e_2}2\Big)}{\e_1+\e _2}+\frac {\phi(\e _1)-\phi\Big(\frac {\e_1-\e_2}2\Big)}{\e_1+\e _2}\Big)\Big)\right)_{\vert_{\vec \e=\vec 0}}
\end{eqnarray*}
for any $a_1, a_2\in \ZZ _{\ge 0}$. Hence the equality of the following holomorphic functions holds
$$\frac {\phi(\e _1)-\phi\Big(\frac {\e_1-2\e_2}3\Big)}{\e_1+\e _2}+\frac {\phi(\e _1)- \phi\Big(\frac {2\e_1-\e_2}3\Big)}{\e_1+\e _2}=\frac {\phi(\e _1)-\phi\Big(\frac {\e_1-\e_2}2\Big)}{\e_1+\e _2}+\frac {\phi(\e _1)-\phi\Big(\frac {\e_1-\e_2}2\Big)}{\e_1+\e _2}
$$
which would imply that
$$\phi\Big(\frac {\e_1-2\e_2}3\Big)+\phi\Big(\frac {2\e_1-\e_2}3\Big)=2\phi\Big(\frac {\e_1-\e_2}2\Big).$$  But this  does not hold for the function $\phi$ as in
Eq.~(\ref{eq:phi}). Note however, that the Taylor expansions at zero agree on either side up to order $1$.
\end{proof}
\subsection{ Discussions and outlook }
 There are by now various  renormalization methods   to evaluate multiple zeta values at non-positive integers, which all use an algebraic Birkhoff Hopf factorization, namely
\begin{enumerate}
\item the present geometric approach which uses a heat-kernel type regularization on the summands  (here polynomials) and a coalgebra on the domains (here cones),
\item the  analytic approach adopted in  \cite{MP} which uses a zeta type regularization as well as a coalgebra on the summands given by tensor products of pseudodifferential symbols,
\item the  number theoretic approach adopted in   \cite{GZ}  to the renormalization of multiple zeta values at non-integer arguments, which uses a zeta type regularization as well as a coalgebra on the summands given by functions $(x_1,\cdots, x_k)\mapsto x_1^{-s_1}\cdots x_1^{-s_k}$.
\end{enumerate}
 There are also other methods which do not use algebraic Birkhoff Hopf factorization such as the approach adopted in \cite{Sa}
 based on  a   formula   expressing the polynomial integrand $Q$ in terms of the integrated polynomial $P(a) =
\int_{
[0,1]^n} Q(a + t)dt$  or   yet a different approach in \mcite{FKMT}   based on Mellin-Barnes integrals to desingularize the  multiple zeta functions.
 The diversity of the existing approaches calls for the need  to relate them conceptually, yet a holy grail for us at this stage. Understanding the relation among these approaches would be a step towards a better understanding of the renormalization group in this context. Whether, from one of those methods, one can by means of a mere change of regularization obtain the renormalized values derived by another method is a first question we hope to address in some future work.

\smallskip

\noindent
{\bf Acknowledgements:} This work was supported by the National Natural Science Foundation of
China (Grant No. 11071176, 11221101 and 11371178) and the National Science Foundation of US (Grant No. DMS 1001855). The authors thank the hospitalities of the Kavli Institute for Theoretical Physics China and Morningside Center of Mathematics. The second author thanks Sichuan University, the Lanzhou University and Capital Normal University for their kind hospitality. She is especially grateful to the group of students at the Lanzhou University to whom she delivered lectures on parts of this survey, for their questions and comments that helped her improve this presentation.

\end{document}